\documentclass[11pt]{article}
\usepackage{amssymb, latexsym, mathtools, amsthm}
\begingroup
    \makeatletter
    \@for\theoremstyle:=definition,remark,plain\do{%
        \expandafter\g@addto@macro\csname th@\theoremstyle\endcsname{%
            \addtolength\thm@preskip\parskip
            }%
        }
\endgroup
\usepackage{enumerate}
\usepackage{pgf,tikz}
\usepackage{pdfsync}
\usepackage{txfonts}
\usepackage{booktabs}
 \usepackage{graphicx}
\definecolor{dnrbl}{rgb}{0,0,0.3}
\definecolor{dnrgr}{rgb}{0,0.3,0}
\definecolor{dnrre}{rgb}{0.5,0,0}
\usepackage[colorlinks=true, citecolor=dnrgr, linkcolor=dnrre, urlcolor=dnrre] {hyperref}
\usepackage{xcolor}
\usepackage{parskip}
\usepackage{tabularx, colortbl}
\usepackage{geometry}
\usepackage{etoolbox}
 \usepackage[scriptsize, up]{caption}
\usepackage{pgf,tikz}
\usetikzlibrary{decorations, decorations.pathmorphing}
\usetikzlibrary {shapes}
\theoremstyle{plain}
\newtheorem{thm}{Theorem}[section]

\newtheorem*{conjecture}{Conjecture}
\newtheorem{lem}[thm]{Lemma}
\newtheorem{coro}[thm]{Corollary}

\theoremstyle{definition}
\newtheorem{rem}[thm]{Remark}
\newtheorem{defi}[thm]{Definition}
\makeatletter

\let\c@table\c@figure
\makeatother

\DeclarePairedDelimiter{\btuple}{\big\langle}{\big\rangle}
\newcommand{\Nat}{\mathbb{N}}

\newcommand{\restr}{\upharpoonright}  
\newcommand{\de}{\downarrow} 

\DeclarePairedDelimiter{\tuple}{\langle}{\rangle}

\newcommand{\bigo}[1]{\mathop{\bf O}\/\left({#1}\right)}
\newcommand{\bigos}[1]{\mathop{\bf O}\/\big({#1}\big)}
\newcommand{\smo}[1]{\mathop{\bf o}\/\left({#1}\right)}
\newcommand{\smos}[1]{\mathop{\bf o}\/\big({#1}\big)}
\DeclarePairedDelimiter{\dbra}{\llbracket}{\rrbracket}

\newcommand{\wgt}[1]{\mathop{\mathtt{wgt}}\/\left({#1}\right)}

\newcommand{\wgtb}[1]{\mathop{\mathtt{wgt}}\/\big({#1}\big)}

\newcommand{\ml}{Martin-L\"{o}f }

\newcommand{\KG}{Ku\v{c}era-G{\'a}cs\ }

\newcommand{\eg}{e.g.\ }
\newcommand{\ie}{i.e.\ }
\newcommand{\ce}{c.e.\ }

\newcommand{\lce}{left-c.e.\ }
\newcommand{\rce}{right-c.e.\ }

\newcommand{\pf}{prefix-free }

\renewenvironment{abstract}
 { \normalsize
  \list{}{
    \setlength{\leftmargin}{.0cm}%
    \setlength{\rightmargin}{\leftmargin}%
    }%
  \item {\bf \abstractname.} \relax}
 {\endlist}
\usepackage{titlesec}
 \titleformat{\paragraph}
{\normalfont\normalsize\em\bfseries}{\theparagraph}{1em}{$\blacktriangleright$\hspace{0.2cm}}
\titlespacing*{\paragraph}
{0pt}{3.25ex plus 1ex minus .2ex}{0.5ex plus .2ex}
\title{Compression of data streams down to their information content
\thanks{Barmpalias was supported by the 
1000 Talents Program for Young Scholars from the Chinese Government No.~D1101130, 
NSFC grant No.~11750110425 and Grant No.~ISCAS-2015-07 from the Institute of Software.
Lewis-Pye was supported by a Royal Society University 
Research Fellowship.}}

\author{George Barmpalias  \and Andrew Lewis-Pye}
\date{\today}
\begin{document}
\maketitle
\begin{abstract}
According to Kolmogorov complexity, 
every finite binary string is compressible to a shortest code -- its information content --
from which it is effectively recoverable. We investigate the extent to which this
holds for infinite binary sequences (streams). We devise a new coding method which
uniformly codes every stream $X$ into an algorithmically random stream $Y$, in such a way that the first
$n$ bits of $X$ are recoverable from the first $I(X\restr_n)$ bits of $Y$, where $I$
is any partial computable information content measure 
which is defined on all prefixes of $X$, and where $X\restr_n$ is the initial segment of $X$ of length $n$.
As a consequence, if $g$ is any computable upper 
bound on the initial segment \pf complexity of $X$, then $X$ is computable from an algorithmically
random $Y$ with oracle-use at most $g$. Alternatively 
(making no use of such a computable bound $g$) one can achieve an oracle-use  bounded above by $K(X\restr_n)+\log n$. 
This provides a strong analogue
of Shannon's source coding theorem for algorithmic information theory.
\end{abstract}
\vspace*{\fill}
\noindent{\bf George Barmpalias}\\[0.5em]
\noindent
State Key Lab of Computer Science, 
Institute of Software, Chinese Academy of Sciences, Beijing, China.\\[0.2em] 
\textit{E-mail:} \texttt{\textcolor{dnrgr}{barmpalias@gmail.com}}.
\textit{Web:} \texttt{\textcolor{dnrre}{http://barmpalias.net}}\par
\addvspace{\medskipamount}\medskip\medskip
\noindent{\bf Andrew Lewis-Pye}\\[0.5em]  
\noindent Department of Mathematics,
Columbia House, London School of Economics, 
Houghton Street, London, WC2A 2AE, United Kingdom.\\[0.2em]
\textit{E-mail:} \texttt{\textcolor{dnrgr}{A.Lewis7@lse.ac.uk.}}
\textit{Web:} \texttt{\textcolor{dnrre}{https://lewis-pye.com}}
\vfill 
\vspace{0.3cm}
\thispagestyle{empty}\clearpage
\section{Introduction} 
A fruitful way to quantify the complexity of a finite object such as a string $\sigma$ in a finite alphabet\footnote{In the
following we restrict our discussion to the binary alphabet, but our results hold 
in general for any finite alphabet.}
is to consider the length of the shortest binary program which prints $\sigma$. 
This fundamental idea gives rise to a theory of algorithmic information and compression, which is based
on the theory of computation and was pioneered by Kolmogorov \cite{MR0184801}
and Solomonoff \cite{Solomonoff:64}. 
The Kolmogorov complexity of a binary string $\sigma$ is the length of the shortest program that
outputs $\sigma$ with respect to a fixed universal Turing machine.
The use of \pf machines in the definition of Kolmogorov complexity was pioneered by
Levin \cite{Lev74tra} and Chaitin \cite{MR0411829}, and 
allowed for 
 the development of a robust theory
of incompressibility and algorithmic randomness for streams (\ie infinite binary
sequences). 
Information content measures, defined by Chaitin \cite{chaitinincomp} after Levin \cite{Lev74tra},
are functions that assign a positive integer value to each binary string, representing
the amount of information contained in the string.

\begin{defi}[Information content measure] \label{X5aVfWRM87}
A partial function $I$ from strings to $\Nat$ is an information content measure if
it is right-c.e.\footnote{A function $f$ is \rce if it is computably approximable from above, \ie
it has a computable approximation $f_s$ such that $f_{s+1}(n)\leq f_{s}(n)$ for all $s,n$.} and $\sum_{I(\sigma)\de} 2^{-I(\sigma)}$ is finite. 
\end{defi}
Prefix-free Kolmogorov 
complexity can be characterized as the minimum (modulo an additive constant)
information content measure.
If $K(\sigma)$ denotes the \pf Kolmogorov complexity of the string $\sigma$, 
then for $c\in \mathbb{N}$ we say that $\sigma$ is {\em $c$-incompressible} if $K(\sigma)\geq |\sigma|-c$.
It is a basic fact concerning Kolmogorov complexity that for some universal constant $c$:  
\begin{equation}\label{fEFIMi19Nt}
\parbox{12cm}{every string $\sigma$ 
has a shortest code $\sigma^{\ast}$ which is
itself $c$-incompressible.}
\end{equation}
Our goal is to investigate the extent to which the above fact holds in an infinite setting, 
\ie for streams instead of strings.
%
In the context of Kolmogorov complexity, algorithmic randomness is 
defined as incompressibility.\footnote{The standard notion of algorithmic randomness for 
streams is due to \ml \cite{MR0223179} and is based
on effective statistical tests. Schnorr \cite{Schnorr:75,MR0354328} 
showed that a binary stream is \ml random if and only if
there exists  $c\in\Nat$ for which all its initial segments are $c$-incompressible.}
So \eqref{fEFIMi19Nt} can be read as follows:
{\em we can uniformly code each string $\sigma$ into an algorithmically 
random string of length $K(\sigma)$.}
In order to formalise an infinitary analogue of this statement, we need to
make use of oracle-machine computations, and work with oracle utilization rather than lengths for codes. 

\begin{defi}[Oracle-use] \label{orau} 
For a binary stream  $X$, we let $X\upharpoonright_n$  
denote the initial segment of $X$ of length $n$. 
Given two binary streams $X$ and $Y$, we say $X$ is computable from $Y$ with 
{\em oracle-use  $n\mapsto f(n)$} if there exists an oracle Turing machine which, 
when given oracle $Y$ and input $n$,  halts and outputs $X\upharpoonright_n$ 
 after performing a computation
in which the elements of $Y$ less than $f(n)$ are queried.
\end{defi}

Our first result states that, if $I$ is any partial computable information content measure $I$,
then every stream $X$ along which $I$ is defined can be compressed into a stream $Y$, 
in such a way that the first $n$ bits of $X$ 
are recoverable from the first $I(X\restr_n)$ bits of $Y$.

\begin{thm}\label{yMWDPEuwPx}
Suppose $I$ is a partial computable information content measure.  Then
every binary stream $X$ satisfying the condition that $\forall n\ I(X\restr_n)\de$ 
can be coded into a \ml random binary stream $Y$, in  such a way that $X$ is computed
from $Y$ with oracle-use $n\mapsto \min_{i\geq n} I(X\restr_i)$.
\end{thm}

Theorem \ref{yMWDPEuwPx} holds, moreover, in an 
entirely uniform fashion, in the sense that there exist a universal constant $c$ and 
a single Turing functional which  computes each $X$ from its code $Y$ with use at most  $n\mapsto c+ \min_{i\geq n} I(X\restr_i)$. 
Prefix-free complexity is not computable, but if we have a computable upper bound on the
initial segment complexity of $X$, the following consequence (derived later from 
Theorem \ref{yMWDPEuwPx}) is applicable.

\begin{coro}\label{FxNsrDIzhj}
If $g$ is a computable upper bound on the initial segment \pf complexity of a stream $X$,
then $X$ is computable from a \ml random stream $Y$ with oracle-use 
$n\mapsto \min_{i\geq n} g(i)$.
\end{coro}
For example, if the \pf complexity of $X$ is bounded above by $5\log n$, we can 
compute $X$ from some algorithmically random $Y$
with oracle-use $5\log n$.
When no useful computable upper bound on the initial segment complexity of the source $X$ 
is known, one can instead apply the  following
 theorem, which gives an upper bound on the oracle-use in terms of $K$ 
(and which, as for Theorem \ref{yMWDPEuwPx}, holds in an entirely 
uniform fashion relative to a fixed universal constant). 
Throughout this paper, logarithms are given base 2.  

\begin{thm}\label{Z6AicbrzJm}
Every binary stream $X$ can be coded into a \ml random binary stream $Y$ such that $X$ is computable from $Y$ 
with oracle-use at most $n\mapsto \min_{i\geq n} (K(X\restr_i)+\log i)$.
\end{thm}

In \S\ref{xbu9Sgt5wv} we explain that our results
offer considerably improved compression in comparison with
the existing methods and are optimal in a strong sense.
Here we briefly outline the main points of our contribution.

{\bf Outline of our contribution compared to the state-of-the-art.}
Our contribution is two-fold: 
first in terms of a considerable improvement on 
the oracle-use to essentially optimal bounds, and second in terms of
a new coding method that is necessary to achieve this result.
The oracle-use we obtain in Theorem \ref{Z6AicbrzJm} is optimal modulo $3\log n$, 
in contrast with the previously best overhead of
more than $\sqrt{n}\cdot \log n$. Given that a typical compressible stream may have  initial segment complexity 
$\smos{\sqrt{n}\cdot \log n}$, even logarithmic  or poly-logarithmic, our results shave-off an overhead  
which is overwhelming compared to the number of bits of the oracle  
that are necessary for the computation of the first
$n$ bits of the source (\ie its Kolmogorov complexity modulo a logarithm).
Even in the worst case of incompressible sources, their initial segment complexity is never more than $n+2\log n$, 
so the previous overhead $\sqrt{n}\cdot \log n$ is still considerable compared to the information coded,
while our overhead $\log n$ from Theorem \ref{Z6AicbrzJm} is exponentially smaller in the same comparison, hence
negligible. In addition, given any computable function $g$, 
in the case where we are interested in coding every
stream of initial segment complexity  at most $g$, our Corollary \ref{FxNsrDIzhj} gives overhead 0
(\ie oracle-use {\em exactly} $n\mapsto g(n)$) compared to the overhead  $\sqrt{n}\cdot \log n$
(\ie oracle-use $n\mapsto g(n)+\sqrt{n}\cdot \log n$)
that is present in all previously known coding methods.

Equally importantly, it is known that the overhead $\sqrt{n}\cdot \log n$
is inherent in any of the previous coding methods, so in order to achieve our optimal bounds it was necessary to
invent a new coding method. Our results are based on a 
new tool, the {\em layered Kraft-Chaitin coding}, which allows for the
construction of infinitary on-line codes with negligible overhead. 
This new methodology is a strong infinitary analogue of
the classic Kraft-McMillan and Huffman 
tools \cite{Kraft:49,McMillan1056818,Huffman4051119} 
for the construction of finitary prefix codes with minimum redundancy, which are part of any
information theory textbook, \eg \cite[Chapter 5]{cover2006}.
Intuitively, our method allows to code several concatenated messages in a binary stream in an on-line
manner, {\em without the need of out-of-band markers or the overheads produced by 
concatenating \pf codes}.\footnote{By \cite{codico16}, the overhead
$\sqrt{n}\cdot \log n$ found in previous methods, is the accumulation of the smaller overheads that
are inherent in  \pf codes, and is the result of concatenating a \pf code for the construction a block-code
of the source. Here and thereafter, the term {\em on-line} refers to the uniform production of {\em approximations} to the code stream from the source.
The actual final code for the source will not be effectively obtainable from the source.}
Given the wide applicability of \pf codes,  our methodology 
is likely to have further applications.

Here we stress that in our coding method,
the code $Y$ is not effectively obtainable from the source $X$.
However the code $Y$ can be effectively approximated, given the source $X$.
On the other hand, the decoding (calculating $X$ from $Y$) is completely effective.

\begin{rem}[Stream compression in historical context]
The archetype of stream coding in classical information theory is
Shannon's source coding theorem (or noiseless coding theorem)
from \cite{shannon48} which assumes a probabilistic source 
and compresses such that the code rate is arbitrarily close to its Shannon entropy.
 When the source is not probabilistic, the problem of
data compression has been traditionally called {\em combinatorial source coding} (\eg see
Ryabko \cite{Rya84,Rya86}). Kolmogorov \cite{MR0184801} elaborated on the differences
between the probabilistic and the combinatorial approaches to 
information theory, and also introduced a 
third algorithmic approach, which set the foundations of algorithmic information
theory (along with Solomonoff \cite{Solomonoff:64}). 
Ryabko \cite{Rya84,Rya86} was one of the first who connected combinatorial coding with 
Kolmogorov complexity, proving an analogue of Shannon's source coding theorem, and his
results are discussed in detail in the following sections. In this sense, our results
can also be seen as analogues of Shannon's source coding theorem in terms of 
algorithmic information,
while the main result of Ryabko \cite{Rya84,Rya86} can be described as an analogue in terms
of Hausdorff dimension.\footnote{We elaborate on the background of the algorithmic approach to 
stream coding in \S \ref{wrsYJDwjf}.}
\end{rem}

\begin{rem}[Probabilistic and algorithmic stream coding]\label{w4hupePuRP}
Many connections between Shannon entropy and Kolmogorov complexity 
have been established in the literature \cite{Leung-Yan-CheongC78, vitanyigrunwald}.
For example,  for any computable probability distribution, 
the expected value of Kolmogorov complexity equals its Shannon entropy, up to a constant.
However the two information measures are conceptually different, with Shannon
entropy assigning complexity to random variables and Kolmogorov complexity focusing on
the complexity of individual finite objects such as strings. This conceptual difference
is also present in the coding theorems. The Shannon source coding theorem
focuses on the average coding rate, i.e.\ optimizing the compression of the typical
streams. In contrast, the algorithmic approach aims at compressing non-random streams,
\ie streams whose initial segments can be described by shorter programs.
The ultimate goal here is to devise a coding process which compresses {\em every} stream 
at a rate that reflects the information content of its prefixes, \ie its initial segment complexity.
The focus in such a universal process
is on non-typical streams which have compressible initial segments,
so the coding can potentially make them recoverable from streams
with oracle-use that matches the information content of their initial segments.
\end{rem}


{\bf Outline of the presentation.}
The goal of this work is to obtain an optimal method for compressing 
binary streams into algorithmically random streams.\footnote{We are not concerned 
with the code streams being effectively constructable from their sources. On the other hand
the decoding should be effective, \ie the source should be computable from the code.}
The first aspect of this goal is the
compression of binary streams
and its relation with Kolmogorov complexity, and is discussed in
\S \ref{wrsYJDwjf}. The second aspect is the problem 
of coding non-trivial information
into algorithmically random strings or streams\footnote{A dual topic is what is known
as randomness or dimension extraction which, roughly speaking, asks for the effective
transformation of a given stream $X$ into a stream $Y$ which is algorithmically random or, at least,
has higher Hausdorff dimension than $X$. Although this problem is only tangential to our topic,
it is very related to the work of Ryabko \cite{Rya84,Rya86} and later Doty \cite{Doty08} which we
discuss in the following. For more information on this we suggest Doty \cite[\S 4]{Doty08}
and the more recent Miller \cite{dimiller}.\label{Sm3tYlsiw2}}
and is discussed in
\S \ref{6CMdw8OhFH}.
We elaborate on the well-known fact that, in a finite setting, maximal compression
gives a natural example of computation from algorithmically random strings, and discuss
the extent to which this phenomenon has been established in an infinitary setting.
In \S \ref{xbu9Sgt5wv} we describe how our results provide optimal answers to
the combined problem of compression and computation from algorithmically 
random oracles in the case of binary streams. Moreover we explicitly compare
our results with the state-of-the-art in the literature, explain why obtaining optimal bounds
required a considerably new methodology, and break down the novelty of our coding into one
defining property.
The main and technical part of our contribution starts in  
in \S \ref{iKCR6nBLqB}, where we develop and verify a sophisticated tool, 
 {\em the layered Kraft-Chaitin theorem}, which allows for the construction of infinite
optimal codes. Then in \S 
\ref{pm7mWNnuPr} we apply our general coding result  
result in the specific setting of a universal discrete semi-measure
corresponding to the underlying optimal \pf machine, in order to obtain the 
theorems discussed at the beginning of this section.
Finally in \S\ref{7vgsUJbYd} we present some concluding thoughts on the present work,
including ideas for possible extensions of our results and some related open problems. 

\subsection{The online compression of binary streams and Kolmogorov complexity}\label{wrsYJDwjf}
Compressibility of strings is well-understood in terms of Kolmogorov complexity.
When we apply the same methodology to a binary stream $X$,
we are interested in the initial segment complexity $n\mapsto K(X\restr_n)$,
and in particular the rate of growth of this function. This is, however,  a {\em non-uniform}
way to look at the compressibility of a stream $X$, since the individual programs
that compress the various initial segments of $X$ down to their initial segment complexity
may be unrelated. Kobayashi \cite{Kobayashi_rep} proposed the following uniform notion of compressibility
for streams.\footnote{Kobayashi's uniform notion of compressibility has proved useful in many contexts,
\eg \cite{MR1455141,Balcazar:2006,koba_rod}. Balc{\'a}zar,  Gavald{\`a}, and  Hermo \cite{MR1455141}, using different
methods than the method in the present article, showed
that any stream with logarithmic initial segment complexity is $\bigo{\log n}$-compressible
in the sense of Definition \ref{8scITRkZc3}.} 

\begin{defi}[Kobayashi \cite{Kobayashi_rep}]\label{8scITRkZc3}
Given a function $f:\Nat\to\Nat$, we say $X$ is $f$-compressible if there exists
$Y$ which computes $X$ with oracle-use bounded above by $f$.
\end{defi}

Ryabko \cite{Rya84,Rya86} discussed an online form of block-coding, which
codes any $X$ into some $Y$ so that the first $u_n$ bits of $Y$
code  $X\restr_n$, for a certain non-decreasing function $n\mapsto u_n$.
He defined the {\em cost of the code on $X$} to be
$\liminf_n u_n/n$ (this was also called the {\em decompression ratio} in 
Doty \cite{Doty08}) and  constructed a universal (in the sense that it applies to any given stream $X$) 
compression algorithm which codes any $X$ into some $Y$ with {\em cost}
the effective Hausdorff dimension of $X$. 
By Mayordomo \cite{Mayordomo02}, the  effective Hausdorff dimension of $X$ is also 
known to equal $\liminf_n K(X\restr_n)/n$.
In terms of Definition \ref{8scITRkZc3}, Ryabko thus showed that every binary stream $X$ is
$f$-compressible, for a non-decreasing function $f$ such that 
\begin{equation}\label{F6Ihiv38MU}
\liminf_n \frac{f(n)}{n} =\liminf_n \frac{K(X\restr_n)}{n}.
\end{equation}
Doty \cite{Doty08}, building on and improving Ryabko's work, explored the above characterization
of effective Hausdorff dimension  in terms of the cost of the optimal compression  in
various resource-bounded settings. In both \cite{Rya84,Rya86} and
\cite{Doty08} the authors ignore sub-linear $\smo{n}$ differences in the oracle-use function $f$
of their coding, and the statements of their results are solely concerned with achieving
the asymptotic \eqref{F6Ihiv38MU}.
An analysis of their arguments, however, shows that for each 
$X$ the oracle-use $f(n)$ is at best $K(X\restr_n)+\sqrt{n}\cdot\log n$.
The overhead $\sqrt{n}\cdot\log n$ is constantly present, independently of the complexity of the source $X$,
and is due to the fact that the coding used is an adaptation of the block-coding method of  G\'{a}cs \cite{MR859105}.
In \S \ref{xbu9Sgt5wv} we elaborate on the limitations of this approach compared to the present work,
and explain why the overhead $\sqrt{n}\cdot\log n$ is severe in the case when the source $X$ is
compressible, which is the focus of algorithmic stream coding as discussed in Remark \ref{w4hupePuRP}.
 
{\bf Explanation of the $\sqrt{n}\cdot\log n$ bottleneck of G\'{a}cs.}
The coding methods discussed above can all be seen as derivatives of
the method of G\'{a}cs \cite{MR859105}, and this is the reason why they all have the
characteristic bottleneck  $\sqrt{n}\cdot\log n$.
G\'{a}cs' method was originally introduced in terms of effectively closed sets, while
Merkle and Mihailovi\'{c} \cite{jsyml/MerkleM04} presented it in terms of martingales.
Here we present a version of G\'{a}cs' method in terms of shortest descriptions 
(to our knowledge, the first in the literature), which codes
each stream $X$ into some $Y$ with oracle-use $K(X\restr_n)+\bigos{\hspace{-0.1cm}\sqrt{n}}$. 
The reason why the overhead
$\bigos{\hspace{-0.1cm}\sqrt{n}}$ is smaller than the original 
$\sqrt{n}\cdot\log n$ is that we do not ensure that the code $Y$ is algorithmically random, which
is a requirement in  \cite{codico16}.
However this original example provides a simple explanation of the main factor 
$\hspace{-0.1cm}\sqrt{n}$ in  G\'{a}cs' bottleneck.

First, we break the source $X$ into successive segments $(\sigma_i)$ so that $|\sigma_i|=i$.
Fix a universal optimal \pf machine $U$ that can also work with finite oracles,
and consider the \pf Kolmogorov complexity $\sigma\mapsto K(\sigma)$ 
and its conditional version $(\sigma,\tau)\mapsto K(\sigma\ |\ \tau)$ with respect to $U$.
For each $i$, we recursively define the shortest program $\sigma^{\ast}_i$ of
$\sigma_i$ as follows. Assuming that $\sigma^{\ast}_j, j<i$ have been defined,
let $\sigma^{\ast}_i$ be a shortest program of $\sigma_i$ with respect to $U$ and
{\em relative, \ie conditional to} the programs $\sigma^{\ast}_j, j<i$, so that
$|\sigma^{\ast}_i|=K(\sigma_i\ |\ \sigma^{\ast}_j, j<i)$.  
Here note that from the concatenation $\sigma^{\ast}_0\ast\cdots\ast\sigma^{\ast}_k$
we can effectively compute the set $\sigma^{\ast}_j, j<k$ due to the fact that $U$ is a
\pf machine. Then given $\sigma^{\ast}_j, j<k$ we can compute $\sigma_j, j<k$
since these programs are the required finite oracles for the \pf machine $U$.
By the same observation it follows that the set $\sigma^{\ast}_j, j<i$ in the conditional part of
the \pf complexity may be viewed equivalently as the string 
$\sigma^{\ast}_0\ast\cdots\ast\sigma^{\ast}_k$, in case $U$ can only work with a single string as an oracle.
 
The required code $Y$ of the source $X$ is the stream $\sigma^{\ast}_0\ast\sigma^{\ast}_1\ast\cdots$.
It remains to show that the first $n$ bits of $X$ can be computed from the first
$K(X\restr_n)+\bigos{\hspace{-0.1cm}\sqrt{n}}$ 
many  bits of $Y$. 
Let $k_n$ be the least number such that $\sum_{i<k_n} |\sigma_i|\geq n$, \ie the least number of
blocks that we need for the calculation of the first $n$ bits of $X$.
Since $|\sigma_i|=i$ for each $i$, we have
$|\sigma^{\ast}_0\ast\cdots\ast \sigma^{\ast}_{k_n}|\approx k_n^2$. In other words,
the first $n$ bits of $X$ are contained within the first 
$\bigos{\hspace{-0.1cm}\sqrt{n}}$ segments $\sigma_j, j\in\Nat$.
Hence in order to show that the
first $n$ bits of $X$ can be computed from the first
$K(X\restr_n)+\bigos{\hspace{-0.1cm}\sqrt{n}}$ 
many  bits of $Y$ it suffices to show that 
\begin{equation}\label{q8VAJEkZbe}
|\sigma^{\ast}_0\ast\cdots\ast \sigma^{\ast}_{k_n}|=
K(\sigma^{\ast}_0\ast\cdots\ast\sigma^{\ast}_{k_n})+\bigo{{k_n}}
\end{equation}
If we let $K(\rho_0,\cdots, \rho_{k-1})$ denote the \pf complexity of the ordered set of strings,
$\rho_j, j<k$, by the symmetry of information \cite{Gacs:74} we have
\begin{equation}\label{LZzFYQfkVf}
|\sigma^{\ast}_0\ast \sigma^{\ast}_1|=
K(\sigma^{\ast}_0)+K(\sigma^{\ast}_1\ |\ \sigma^{\ast}_0)+\bigo{1}=
K(\sigma^{\ast}_0,\sigma^{\ast}_1)+\bigo{1}=
K(\sigma^{\ast}_0\ast\sigma^{\ast}_1)+\bigo{1}.
\end{equation}
By iterating this argument, for each $n>0$ we get
\[
K(\sigma^{\ast}_0)+K(\sigma^{\ast}_1\ |\ \sigma^{\ast}_0)+\cdots+ 
K(\sigma^{\ast}_k\ |\ \sigma^{\ast}_j, j<k_n)=
K(\sigma^{\ast}_0,\cdots,\sigma^{\ast}_{k_n})+\bigo{k_n}.
\]
so \eqref{q8VAJEkZbe} holds and the oracle-use 
in the computation of $X$ from $Y$ is $K(X\restr_n)+\bigos{\hspace{-0.1cm}\sqrt{n}}$ as required. 

Hence, at least with the current choice of block-lengths, 
the only way that the overhead $\bigos{\hspace{-0.1cm}\sqrt{n}}$ 
in the above argument can be reduced
is if the constant overhead $\bigo{1}$ of the symmetry of 
information principle in \eqref{LZzFYQfkVf} can be eliminated completely,
\ie can be made 0 (at least with respect to {\em some} universal \pf machine).
It is not hard to show, and it is widely known, that this is impossible.
In the following we show why opting for different block-lengths than G\'{a}cs' choice of $|\sigma_i|=i$
can only increase the overhead.

{\bf Why different block-lengths do not reduce the overhead in G\'{a}cs' coding.}

As discussed above,
the computation of $X\restr_n$ requires
the segment $\sigma^{\ast}_0\ast\cdots\ast \sigma^{\ast}_{k_n}$. 
Taking into account the overheads as before (one for each block), for arbitrary block-lengths $|\sigma_i|$,
\begin{equation}\label{bCyu4miz4T}
\textrm{we need oracle-use $n+\bigo{k_n}+|\sigma_{k_n}|$.}
\end{equation}
The secondary overhead $|\sigma_{k_n}|$ in this calculation
was not mentioned under the case when $|\sigma_i|=i$ because in this case it is $k_n$
so it can be absorbed in $\bigo{k_n}$.
The intuition for the added overhead $|\sigma_{k_n}|$
is that,
due to the fact that $n$ could be slightly larger than
$\sum_{i< k_n} |\sigma_i|$, and since we have chosen to block-code $X$, for the calculation of $X\restr_n$
we need the last block $\sigma_{k_n}$ which may contain bits of $X$ that are larger than $n$.
Hence there is a trade-off between the length of blocks and the size of the original overhead
(which depends on the number of blocks below $n$): longer blocks reduce the original overhead
(less blocks in $X\restr_n$) but increase the secondary overhead $|\sigma_{k_n}|$ in the above calculation.
Smaller blocks reduce $|\sigma_{k_n}|$ but increase the number of blocks $k_n$ in $X\restr_n$, 
hence the original overhead.
By choosing  $|\sigma_i|=i$, the oracle-use for $X\restr_n$ becomes 
$n+\bigo{k}+|\sigma_{k_n}|=n+\bigo{k_n}\approx \bigos{\hspace{-0.1cm}\sqrt{n}}$.

If we choose smaller blocks than $|\sigma_i|=i$ for each $n$ we will have more blocks in $X\restr_n$
so the primary overhead $\bigo{k_n}$ in \eqref{bCyu4miz4T} can only increase.
On the other hand, small increases such as $|\sigma_i|=2i$ do not make any difference since,
for example, the secondary overhead $|\sigma_{k_n}|$ becomes $2k_n$.
If we choose larger blocks such as $|\sigma_i|=i^2$,  
the number of blocks decreases to about $k_n\approx n^{1/3}$, but then the secondary overhead
$|\sigma_{k_n}|$ in \eqref{bCyu4miz4T} becomes $n^{2/3}$ which is even worse than the 
$\bigos{\hspace{-0.1cm}\sqrt{n}}$ that
we had before. Exponential-sized blocks such as  $|\sigma_i|=2^{i}$ increase the oracle-use even more,
since in this case the number of blocks in $X\restr_n$ are approximately $\log n$ and the
secondary overhead $|\sigma_{k_n}|$ in \eqref{bCyu4miz4T} becomes $2^{\log n}=n$, which is much worse than
the $\bigos{\hspace{-0.1cm}\sqrt{n}}$ that we had with the choice $|\sigma_i|=i$.

We have shown that G\'{a}cs' choice  of  $|\sigma_i|=i$ is essentially optimal with respect to his block-coding method.
A different and more thorough analysis of the limitations of the block-coding of G\'{a}cs,
in his original formulation in terms of effectively closed sets, can be found in 
\cite{codico16}.

\subsection{Coding binary streams into algorithmically random streams}\label{6CMdw8OhFH}
Intuitively speaking, if a string or stream is  sufficiently algorithmically random,
then it should not be possible to extract any `useful' information from it. 
Plenty of technical results that support this intuition have been established
in the literature, for various levels and notions of randomness.\footnote{For example, Stephan 
\cite{MR2258713frank} showed that incomplete \ml random
binary streams  
cannot compute any complete extensions 
of Peano Arithmetic; similar results are presented in 
Levin \cite{Levin:2013:FI}. A simple example showing that sufficiently random strings
cannot be decompressed into anything, is presented in \cite{codico16}.}
In stark contrast,  Ku\v{c}era \cite{MR820784} 
and G\'{a}cs \cite{MR859105} showed that every stream is computable from a \ml random stream (a result that
is now known as the \KG theorem).
Bennett \cite{Bennett1988} views this result as the infinitary 
analogue of the classic fact that every string can be coded into an algorithmically
random string, namely its  shortest 
description. 
Doty \cite{Doty08}, quite correctly, points out that this analogy is missing a rather crucial
quantitative aspect: according to the classic fact, every string $\sigma$ is computable from
a random string $\sigma^{\ast}$ of length $K(\sigma)$, although the coding provided by
 Ku\v{c}era \cite{MR820784} 
and G\'{a}cs \cite{MR859105} leaves much to be desired regarding the number of bits required from the
random oracle in order to recover a given number of bits of the source. 
The analogue of `length' for codes in the infinitary setting is the oracle-use function $n\mapsto f(n)$ that determines the length of the initial segment of $Y$ which is queried during a computation of $X\restr_n$. 
A quantitative version of Bennett's analogy would ask that
every stream $X$ has an algorithmically random code $Y$ which computes it with oracle-use 
close to $n\mapsto K(X\restr_n)$. 
The coding methods of  Ku\v{c}era and G\'{a}cs fall short of facilitating such a strong result in 
two ways:
\begin{enumerate}[\hspace{0.5cm}(a)]
\item the oracle-use is  oblivious to the stream being coded;
\item this uniform oracle-use is considerably higher than the initial segment complexity of any stream;
\end{enumerate}
Clause (b) was the main topic of discussion in Barmpalias and Lewis-Pye \cite{codico16}, where
it was pointed out that Ku\v{c}era's coding gives oracle-use $n\log n$ and G\'{a}cs' refined coding
gives oracle-use $n+\sqrt{n}\log n$. In the same article it was demonstrated that these methods 
(and their generalisations) cannot be extended
in order to give significantly smaller oracle-uses. 



Doty \cite{cie/Doty06,Doty08} tackled this challenge with respect to (a) above, 
by combining the ideas of Ryabko \cite{Rya84,Rya86}
and the coding of  G\'{a}cs \cite{MR859105}, in order to produce an adaptive coding of any stream $X$ into
an algorithmically random stream $Y$, where the oracle-use $f$ in the computation of $X$ from $Y$
reflects the initial segment complexity of $X$, in the sense that \eqref{F6Ihiv38MU} holds.
So Doty provided a way to code any stream $X$ into an algorithmically random stream,
achieving decompression ratio equal to the effective Hausdorff dimension of $X$.
Doty's work provides a quantitative form of Bennett's analogy, but is still a step away
from the direct analogy of obtaining oracle-use close to 
$n\mapsto K(X\restr_n)$ for the computation
of a source $X$ from its random code  -- the requirement on the oracle use $f$ that the \emph{ratio} $f(n)/n$ 
should be asymptotically equal to $K(X\restr_n)/n$ is much weaker.\footnote{For the sake of comparison, if the effective Hausdorff dimension
of $X$ is 1, then Doty's method codes $X$ into a random stream $Y$ which computes $X$
with oracle use $n\mapsto n+\sqrt{n}\log n$, \ie the same as in 
G\'{a}cs \cite{MR859105}. In other words, although Doty \cite{cie/Doty06,Doty08} provides an adaptive coding
where the dimension of the source is reflected in the oracle-use, he does not provide
better worse-case redundancy than G\'{a}cs \cite{MR859105}.} 
A more recent attempt by Barmpalias and Lewis-Pye \cite{optcod16}, using very different methods, 
focused on tackling clause (b)
by producing a coding method (oblivious in the sense of (a) above) 
which achieves logarithmic worse-case redundancy, namely oracle-use 
$n+\epsilon\cdot\log n$ for any $\epsilon>1$. 
Barmpalias, Lewis-Pye and Teutsch \cite{bakugaopt}
showed that this is strictly optimal with respect to oblivious oracle-use, \ie there exist
streams which are not codable into any algorithmically random stream with
redundancy $\log n$.
Despite the simplicity and optimality of this new coding technique, it falls short of
dealing with clause (a) above.

\subsection{Novelty and explicit comparison  with existing work in the literature}\label{xbu9Sgt5wv}
In order to compare our work with the state-of-the-art, recall that the goal achieved in the present work is 
\begin{equation}\label{hm9v8PK19R}
\parbox{14.5cm}{to code all streams $X$ into algorithmically
random $Y$, so  that the length of $Y$ that is required to recover the first
$n$ bits of $X$ is essentially $K(X\restr_n)$, \ie the information contained in $X\restr_n$.}
\end{equation}
In this sense, our work deals with both issues (a) and (b) discussed in \S \ref{6CMdw8OhFH}, 
from which all currently known approaches to coding into
algorithmically random streams suffer.

{\bf Optimality of our results.}
If $X$ is computable from $Y$, then for almost all $n$
the oracle use $g(n)$ must be larger than $K(X\restr_n)-2\log n$. In order to see this, note that
$\lim_n (2\log n -K(n))=\infty$ and 
each $X\restr_n$ can be described with the prefix-free code
that starts with a shortest \pf description of $n$, concatenated with $Y\restr_{g(n)}$.
Precisely speaking, Theorem \ref{Z6AicbrzJm} is tight modulo $3\log n$.

{\bf Comparison with G\'{a}cs-Ryabko-Doty.}
The state-of-the-art result towards \eqref{hm9v8PK19R} was, until now,
Doty \cite{Doty08}, where the oracle-use obtained falls short of the target
 $K(X\restr_n)$ by $\smo{n}$. Unfortunately, the statements of the results in Doty \cite{Doty08}
 do not state the exact value of the error  $\smo{n}$, but an analysis of the proofs shows that
 this is  $\sqrt{n}\cdot\log n$ at best.  Doty's approach is an amalgamation of the method of
Ryabko \cite{Rya84,Rya86} and the block-coding of
G\'{a}cs \cite{MR859105}, which explicitly states an overhead of  $3\sqrt{n}\cdot\log n$.
In \cite{codico16}  it was demonstrated that, although G\'{a}cs' bound can be reduced to
$\sqrt{n}\cdot\log n$ by more careful calculations, no substantial improvement is possible using this approach.
We may conclude that the overhead  $\sqrt{n}\cdot\log n$ is intrinsic in the existing coding methods,
and compare it with the overhead $\log n$ of our Theorem \ref{Z6AicbrzJm}. Given that
even the worst-case possible oracle-use is less than $n+2\log n$, the quantity
$\sqrt{n}\cdot\log n$ that we shave-off from the overhead is considerable.
If we consider sources $X$ which are compressible at a certain rate, \ie
the Kolmogorov complexity of their initial segments is at most $g$
(\eg $7\log n$ or $3(\log n)^2$) then our Corollary \ref{FxNsrDIzhj} gives
oracle-use {\em exactly} $g$ while Doty \cite{Doty08} gives 
$g+\sqrt{n}\cdot\log n$ at best; in terms of overheads, it is $\sqrt{n}\cdot\log n$ versus 0. 
In such situations, the oracle-use provided
by previous methods is overwhelming compared to the actual information that is being coded,
but also overwhelming compared to the oracle-use given by our method.

{\bf Comparison with the worse-case bounds of G\'{a}cs and Barmpalias and Lewis-Pye.}
A universal upper bound on the \pf initial segment complexity of every stream is $n+2\log n$, or even
 $n+\epsilon\cdot \log n$ for any $\epsilon>1$. Hence a weaker version of \eqref{hm9v8PK19R}
 would be to devise a method which codes
each $X$ into an algorithmically
random stream $Y$, in such a way that the number of bits of $Y$ that are required to recover the first
$n$ bits of $X$ is $n+2\log n$, \ie essentially the worse-case initial segment complexity.
This is what we described as {\em oblivious oracle-use} in \S\ref{6CMdw8OhFH}, \ie the oracle-use is
fixed as a function and does not depend on the stream being coded, or its complexity.
The methods of Ku\v{c}era \cite{MR820784}
and G\'{a}cs \cite{MR859105} were of this type, with the first one achieving 
oracle-use $2n$ and the later,  oracle-use 
$n+3\sqrt{n}\log n$, \ie {\em redundancy} $3\sqrt{n}\log n$.
These methods may be described as
 forms of block-coding, in the sense that some
increasing computable 
sequence $(n_i)$ is chosen, which splits  the source stream $X$ into countably many blocks, 
which are coded into corresponding blocks in a code $Y$ (determined by another
increasing sequence $(m_i)$).
Ku\v{c}era \cite{MR820784} choses $n_i=i$ while G\'{a}cs 
\cite{MR859105} considers $m_i\approx i^2$.
In Barmpalias and Lewis-Pye \cite{codico16} it was demonstrated that
these methods cannot give redundancy less than $\sqrt{n}\log n$,
which can be seen as the sum of logarithmic overheads for each block 
on the code $Y$, where the $i$th block has length $i$. 

Recently, Barmpalias and Lewis-Pye \cite{optcod16}  used a different method 
in order to obtain oblivious bounds such as $n+2\log n$ 
(even $n+\epsilon\cdot \log n$ for any $\epsilon>1$) and in \cite{bakugaopt}
it is shown that these worst-case oblivious bounds are optimal (even up to extremely small
differences such as $\log\log\log n$ -- see \cite{optcod16,bakugaopt} for the exact characterization).
We note that these oblivious bounds are also obtained  via our Corollary
\ref{FxNsrDIzhj}, which is however a much stronger result and is based on the 
considerably more sophisticated method of \S\ref{iKCR6nBLqB}.

{\bf A defining aspect of our coding which allows to shave-off  G\'{a}cs'  overhead of $\sqrt{n}\log n$.}
As discussed above, all of the existing coding methods for \eqref{hm9v8PK19R} (or its weaker, worst-case form)
carry an overhead of at least  $\sqrt{n}\log n$ in the oracle-use. The reason for this is that they are all
derivatives of G\'{a}cs' method, which cannot give better bounds as it was demonstrated in  \cite{codico16}.
There is a specific feature in our method of \S\ref{iKCR6nBLqB} and, in a simple form, in  \cite{optcod16},
which is absent from all the above forms of block-coding and
which allows the elimination of this overhead. In any derivative of
G\'{a}cs'  method, the source $X$ and the code $Y$ are split into blocks 
of lengths $n_{i}-n_{i-1},m_{i}-m_{i-1}$ respectively for the $i$th block, where $(n_i), (m_i)$ are increasing 
sequences (possibly depending on $X$), and for each $i$
\begin{equation}\label{ba23HDmLNi}
\parbox{11cm}{ the segment $Y\restr_{m_i}$ of the code is uniquely specified by the
segment $X\restr_{n_i}$ of the source and the set of incompressible sequences of length $\leq m_i$.}
\end{equation}
In other words,  if the coding produces $c$-incompressible codes $Y$ 
(\ie such that $K(Y\restr_i)\geq i-c$ for all $i$) given the segment $X\restr_{n_i}$  that is being coded and
the set of  $c$-incompressible strings  of length $\leq m_i$, we can recover the code
 $Y\restr_{m_i}$.
This uniqueness property is no longer present in our coding method of  \S \ref{iKCR6nBLqB}, 
and this is the defining novel characteristic
which allows for the elimination of the bottleneck $\sqrt{n}\log n$. 

\section{The layered Kraft-Chaitin theorem for infinitary coding}\label{iKCR6nBLqB}
The Kraft-Chaitin theorem, which is an effective version of Kraft's inequality, 
is an indispensable tool for the construction of \pf codes.\footnote{Kraft's inequality is from  
\cite{Kraft:49} and features in many textbooks such as \cite[\S 1.11.2]{Li.Vitanyi:93}.
The Kraft-Chaitin theorem was first used in \cite{Schnorr:73,levinphd,levinthesistran,MR0411829}; 
also see \cite[\S 3.6]{rodenisbook} for a clear presentation and some history.} 
This section is devoted to
what might be thought of as a \emph{nested} version of this classic result, which we call the {\em layered
Kraft-Chaitin theorem}, and which can be used in order to produce infinitary codes. 
 Despite its additional sophistication, the proof of our layered
Kraft-Chaitin theorem is based on similar ideas to the proof of the classic
Kraft-Chaitin theorem. For this reason, we start in \S \ref{6Azke7oXF} by formally stating
certain notions associated with this classic result and its proof. This
proof is based on a particularly succinct presentation in \cite[\S 3.6]{rodenisbook}, 
where it is partially credited to Joseph S.~Miller. Even if the reader is familiar with the Kraft-Chaitin theorem
and its proof, we recommend reading through \S \ref{6Azke7oXF} which introduces terminology
which will be used freely in later sections.

\subsection{Plain Kraft-Chaitin requests and the greedy solution}\label{6Azke7oXF}
The Kraft-Chaitin theorem can be viewed as providing a greedy online algorithm for
satisfying a sequence of \emph{requests}. The satisfaction of each request
 requires that a string be produced of a certain length (as specified by the request), 
 and which is incompatible with all strings used to satisfy previous requests.  

\begin{defi}[Kraft-Chaitin sequence of requests]\label{SjtisRRsER}
A \emph{Kraft-Chaitin} (KC) \emph{sequence} is a finite or infinite sequence of positive integers
$\tuple{\ell_i, i<k}$ where $k\in\Nat\cup\{\infty\}$. 
We say that a sequence $\tuple{\sigma_i, i<k}$ of strings  is a \emph{solution to} 
the KC-sequence $\tuple{\ell_i, i<k}$,
if $|\sigma_i|=\ell_i$, $\sigma_i\neq\sigma_j$ for all $i\neq j$ and the set $\{\sigma_i\ |\ i<k\}$ is prefix-free. 
\end{defi}

In the next section we will define
a more general version of KC-sequences. For this reason,
we also refer to the notion of Definition \ref{SjtisRRsER} as a {\em plain} KC-sequence and its
terms as {\em plain requests}.

\begin{defi}[Weight and trace of a KC-sequence]\label{vgGN7fsYKH}
The weight of a KC-sequence $L=\tuple{\ell_i, i<k}$ is 
$\sum_{i<k} 2^{-\ell_i}$, and is denoted by $\wgt{L}$.  
The trace of $L$ is the binary expansion of $1-\sum_{i< k} 2^{-\ell_i}$, 
as a binary stream or string, depending on whether the length $k$ of the sequence is infinite or finite.
\end{defi}

The Kraft-Chaitin theorem says that every computable  KC-sequence 
 $\tuple{\ell_i, i<k}$  with weight at most $ 1$ has a computable solution. By 
 Kraft's inequality, if the weight of the KC-sequence is more than 1, then it does not have a solution. 
 To give a proof for the theorem we define a 
 {\em greedy strategy}, which  constructs a solution for any 
 KC-sequence with weight at most 1. This strategy enumerates a solution 
 $G_t=\tuple{\sigma_i, i< t}$ to each initial segment of the given KC-sequence of length $t$,
 and is defined inductively on the length of the
 KC-sequence. 
 
 Given a string $\sigma$,  let 
 $\dbra{\sigma}$ denote all binary streams which have $\sigma$ as a prefix; similarly, if $H$ is
 a set of binary strings, we let $\dbra{H}$ be the union of all $\dbra{\sigma}$, $\sigma\in H$.
 The strategy is based on 
 monitoring the trace of the  KC-sequence, while also constructing an auxiliary sequence $(F_t)$ of sets of
 strings such that for each $t$:
 \begin{equation}\label{lnWETnVBbs}
\parbox{13cm}{$\dbra{G_t\cup F_t}=2^{\omega}$, $G_t\cup F_t$ is prefix-free, and
there is a one-to-one map $i_j\mapsto \mu_j$ from the positions of the 1s in the trace of 
$\tuple{\ell_i, i<t}$ onto the set $F_t$ such that $|\mu_j|=i_j$.}
\end{equation}
We call the $F_i$ {\em filler sets}. The intuition is that they represent the strings
which are available to be used in extending the current solution to an initial segment of the KC-sequence
(either the actual strings in $F_i$ or their extensions).
Let $\ast$ denote
the concatenation of strings.
\begin{defi}[Greedy solution]\label{pXJMi4x4i}
Given a KC-sequence $\tuple{\ell_i, i< k}$ with weight at most 1,
the greedy solution $G_k=\tuple{\sigma_i, i< k}$ is defined inductively. 
We define  $G_1=\{0^{\ell_0}\}$ and
$F_1=\{1,01, \dots, 0^{\ell_0-1}1\}$. Assuming that $G_t,F_t$ have been defined 
and satisfy \eqref{lnWETnVBbs}, we define $G_{t+1},F_{t+1}$ as follows. 
Note that there exists a 1 in the trace of  $\tuple{\ell_i, i< t}$ on a position which is
at most $\ell_{t}$, otherwise the weight of  $\tuple{\ell_i, i< t+1}$ would exceed 1.
Let $p$ be the largest such position and consider the string $\mu\in F_t$ which corresponds
to this position according to  \eqref{lnWETnVBbs}. Then:
\begin{itemize}
\item Let $\sigma$ be the leftmost extension of $\mu$ of length $\ell_{t}$, namely
$\mu\ast 0^{\ell_{t}-p}$;
\item If $p<\ell_{t}$ then let $R$ be the set of strings  $\{ \mu\ast 0^{\ell_{t}-p-i}\ast 1: \ 0<i\leq \ell_{t}-p\}$,  and if $p=\ell_t$ then let $R=\emptyset$;
\item Define $F_{t+1}=F_t\cup R-\{\mu\}$ and $\sigma_{t}=\sigma$.
\end{itemize}
Note that  \eqref{lnWETnVBbs}
continues to hold for $t+1$ in place of $t$. This concludes the definition of  $G_{t+1}$.
\end{defi}
If we want to emphasise that the algorithm just described gives a solution to a {\em plain} KC-sequence,
we refer to this solution as the {\em plain} greedy solution (in \S \ref{AGapDZxJHM} we will
discuss the {\em layered} greedy solution).

\subsection{Relativized plain KC-sequences and their greedy solution}\label{f9OYHFR8ER}
All notions  discussed in \S \ref{6Azke7oXF} have straightforward
relativizations as follows.
A KC-sequence $\tuple{\ell_i, i<k}$ relative to a string $\tau$ is defined exactly as in Definition 
\ref{SjtisRRsER}, except that each $\ell_i$ is interpreted as the request `produce an extension
of $\tau$ of length $\ell_i$'. 
Let $\preceq, \prec$ denote the prefix and proper prefix (\ie prefix and not equal to) relations amongst strings.
A solution $\tuple{\sigma_i, i<k}$ to $\tuple{\ell_i, i<k}$ is defined
as in Definition \ref{SjtisRRsER}, with the extra condition that 
$\tau\preceq\sigma_i$ for all $i$.
The weight of a KC-sequence $\tuple{\ell_i, i<k}$ relative to $\tau$ is again given by 
$\sum_{i<k} 2^{-\ell_i}$, the trace is the binary expansion of $2^{-|\tau|}-\sum_{i<k} 2^{-\ell_i}$, 
and the relativized Kraft-Chaitin theorem says that a
KC-sequence $\tuple{\ell_i, i<k}$ relative to $\tau$ has a solution, provided that
its weight is at most $2^{-|\tau|}$.
The greedy solution to 
a KC-sequence $\tuple{\ell_i, i<k}$ relative to $\tau$ is defined exactly\footnote{The
remark `Note that there exists\dots' in  Definition \ref{pXJMi4x4i} now should be 
`Note that there exists \dots, otherwise the weight of  $\tuple{\ell_i, i< t+1}$ would exceed $2^{-|\tau|}$'.} 
as in Definition \ref{pXJMi4x4i} with the only difference that we now start with:
\[
\textrm{$G_1=\{\tau\ast 0^{\ell_0-|\tau|}\}$
\hspace{0.3cm} and \hspace{0.3cm}
$F_1=\{\tau\ast1, \tau\ast01, \dots, \tau\ast0^{\ell_0-1-|\tau|}1\}$.}
\]
As a result, the solutions $G_i$ and the filler sets $F_i$ now consist entirely of
extensions of $\tau$. 
We refer to a plain KC-sequence relative to a string as a
{\em relativized plain KC-sequence} if we want to emphasise the difference with the
notion of Definition \ref{SjtisRRsER}.

\subsection{Layered Kraft-Chaitin requests}\label{TRNDMYrIfL}
Informally speaking, a layered Kraft-Chaitin request could be a
plain request of the form ``produce a string $\sigma$ of length $\ell$'' such as in Definition \ref{SjtisRRsER},
but could also be a nested request of the form ``produce a string $\sigma$ of length $\ell$
which is a proper extension of a string that was used in order to satisfy a certain previous request''.
So a single layered request may actually involve a long sequence of previous requests, 
the length of which determines 
the {\em depth} of the request. Prefix-freeness is required just as in 
Definition \ref{SjtisRRsER}, except that now we only require it layer-wise, \ie amongst
requests of the same depth. A layered request is now represented by a tuple
$(u,\ell)$ whose second coordinate is the requested length, while the  first coordinate is the
index of the previous request that it points to, according to the informal discussion above.

\begin{defi}[Layered Kraft-Chaitin requests]\label{uK47ZAOVAK}
A layered KC-sequence is a finite or infinite sequence $\btuple{r_i=(u_i,\ell_i), i<k}$, 
where $k\in\Nat\cup\{\infty\}$, $\ell_i\in\Nat$,
such that $(u_0,\ell_0)=(\ast,0)$ and for each $i>0$ we have $u_i\in\{ 0, \dots, i-1 \}$ and $\ell_i>\ell_{u_i}$.
The request $r_0=(u_0,\ell_0)$ is said to be the empty request and is the only 0-depth request.
If $i>0$ and request $r_{u_i}$ is a $j$-depth request, then
$r_i$ is a $(j+1)$-depth request. The length of $(u_i,\ell_i)$ is $\ell_i$.
\end{defi}
Given requests $(u_i, \ell_i)$, $(u_j, \ell_j)$ such that $u_j=i$, we say that
$(u_j, \ell_j)$ \emph{points to} $(u_i, \ell_i)$. This relation defines a partial order, a tree, amongst
the layered KC-requests. Note that in Definition \ref{uK47ZAOVAK}
we require that the length $\ell_i$ of each request $(u_i,\ell_i)$ should be strictly larger
than the length of the request that it points to.
The empty request does not have any meaning and it only exists for notational convenience.
The $1$-depth requests all point to the empty request, and 
can be viewed as the plain KC-requests of Definition \ref{SjtisRRsER}.

\begin{defi}[Predecessor and successor requests]\label{iAEntPEpUs}
Given two requests $(u_i,\ell_i), (u_j,\ell_j)$ in a  
layered KC-sequence, we say that $(u_i,\ell_i)$ is an immediate predecessor of $(u_j,\ell_j)$ (and that $(u_j,\ell_j)$ is an immediate successor of $(u_i,\ell_i)$)
if $u_j=i$, \ie if $(u_j,\ell_j)$ points to $(u_i,\ell_i)$. 
\end{defi}
We must formally define what is meant by a solution to a
layered KC-sequence.  According to this definition,
drawing a parallel with Definition \ref{SjtisRRsER}, a layered request $(u_i, \ell_i)$
may be satisfied by \emph{several} 
strings of length $\ell_i$ in the solution. 
The reason for this feature will become clear in the discussion after Definition \ref{etNHm5OGHS}.
Recall that $\preceq,\prec$ denote the prefix and proper prefix relations amongst strings.
\begin{defi}[Satisfaction of layered KC-requests]\label{jPdiYuWBqz}
Suppose that  $\btuple{(u_i,\ell_i), i<k}$
is a  layered KC-sequence and 
for each $t\in\Nat$, let $I_t$ be the set of indices $i$ such that $r_i=(u_i,\ell_i)$ is a $t$-depth request.
We say that  a sequence 
$\tuple{S_{i}, i<k}$ of sets of strings
satisfies, or is a solution to the KC-sequence $\btuple{(u_i,\ell_i), i<k}$, if
for each $i,j<k$ with $i\neq j$:
\begin{enumerate}[\hspace{0.5cm}(a)]
\item $S_i$ consists of strings of length $\ell_i$ and  if $i>0$, 
for every $\sigma\in S_i$ there exists $\tau\in S_{u_i}$ such that  $\tau\prec\sigma$.
\item  if $u_i=u_j$ then for each $\sigma\in S_i$, $\tau\in S_j$ we have
$\sigma\not\preceq\tau$ and $\tau\not\preceq\sigma$.
\end{enumerate}
\end{defi}

Note that in Definition \ref{uK47ZAOVAK}, we require $\ell_i>\ell_{u_i}$.  
Condition (b) in Definition \ref{jPdiYuWBqz} implies that $S_j\cap S_i=\emptyset$ for each
$j\neq i$, when $\tuple{S_{i}, i<k}$ is a solution to a 
layered KC-sequence. Indeed, if
$\sigma\in S_i\cap S_j$, by the monotonicity of the lengths in
 Definition \ref{uK47ZAOVAK} it follows that $i$ is not an ancestor
 of $j$ in the tree of all layered KC-requests, and vice-versa. Hence if
 $x$ is the
greatest common ancestor  of the $i$th and $j$th requests, 
we have $x<i$, $x<j$.
If  $i',j'$ are the unique ancestor requests of $i,j$ respectively that point to $x$,
then  $i\neq j$ implies $i'\neq j'$. 
This contradicts
(b) in Definition \ref{jPdiYuWBqz}. 
By the same argument, if $i\neq j$, the $i$th and the $j$th request
have the same depth and $\sigma\in S_i, \tau\in S_j$ we have
$\sigma\not\preceq\tau$ and $\tau\not\preceq\sigma$.
since $\sigma\in S_i\cap S_j$ would have a prefix in $S_{i'}$ and a prefix in 
 $S_{j'}$. In other words, $\cup_{j\in I_t} S_j$ is prefix-free for each $t\in\Nat$.

 Note also that Definition \ref{uK47ZAOVAK} is a generalisation of  Definition
\ref{SjtisRRsER}. In particular, a layered KC-sequence consisting entirely of (the empty request and) $1$-depth requests
can be identified with a plain KC-sequence. 

\begin{defi}[Weight of a layered KC-sequence]\label{etNHm5OGHS}
Given a layered KC-sequence
$\btuple{(u_i,\ell_i), i<k}$ we define its weight as
$\sum_{i\in (0,k)} 2^{-\ell_i}$.
\end{defi}

In analogy with the classic Kraft-Chaitin theorem, 
we wish to show that if a layered KC-sequence has appropriately bounded weight, then it has
a solution. By the classic Kraft-Chaitin theorem,
every layered KC-sequence consisting entirely of (the empty request and) $1$-depth requests has a solution
\emph{consisting of singletons}, provided that its weight is at most 1. 
This is no longer true, however,  if the  layered KC-sequence
contains deeper requests. Consider, for example,  the 
layered KC-sequence $(\ast, 0),(0, 2), (1,3),(1,3),(1,3),(1,3)$
which has weight $2^{-2}+2^{-1}<1$. Here the second request is a 1-depth request, while the last four
requests are 2-depth requests. A solution to this layered KC-sequence consisting of
singletons would necessarily 
involve a string $\sigma$ of length 2 for the 1-depth request, and four strings of length 3
which extend $\sigma$. Clearly this is impossible, so this
layered KC-sequence  does not have a solution consisting of singletons.
It does, however, have the solution  
$S_0=\{\emptyset\}$, 
$S_1=\{00, 01\}$, 
$S_2=\{000\}$,
$S_3=\{001\}$,
$S_4=\{010\}$,
$S_5=\{011\}$. This is the reason that we allow layered requests to be satisfied by \emph{sets} of strings rather than by individual strings. 

\begin{defi}[Uniform solution]
A uniform solution to a KC-sequence
 $\tuple{(u_i,\ell_i), i<k}$ is a double sequence $(S_i[t])$ of sets of strings
 such that $S_i[t]\subseteq S_i[t+1]$ for all $i,t<k$, and for each $t<k$
 the sequence $\tuple{S_i[t], i<t}$ is a solution to the KC-sequence 
 $\tuple{(u_i,\ell_i), i<t}$. 
\end{defi}

Note that if $(S_i[t])$ is a uniform solution to
$\tuple{(u_i,\ell_i), i\in\Nat}$, then if $S_i:=\lim_t S_i[t]$,
 the sequence $(S_i)$ is a solution to 
$\tuple{(u_i,\ell_i), i\in\Nat}$
in the sense of Definition \ref{jPdiYuWBqz}.
Our goal now is to prove the following theorem.

\begin{thm}[Layered KC-theorem]\label{vOPTnu8Le8}
Every layered KC-sequence of weight at most $1$ has a uniform solution.
If, in addition, the layered KC-sequence is computable, then there exists a uniform solution which is computable. 
\end{thm}

Before describing the proof, we introduce some useful terminology.

\begin{defi}[Characteristic sequence of a layered request]\label{ntlNU6bLF3}
Let $\tuple{(u_{i}, \ell_{i}), i<k}$ be a layered KC-sequence.
The \emph{characteristic sequence} of the empty request is 
$\tuple{v_j, j<1}$ with $v_0=0$ and
the characteristic sequence of a 1-depth request $(u_{i}, \ell_{i})$ is 
$\tuple{v_j, j<2}$ with $v_0=0, v_1=i$.
Inductively assuming that the characteristic sequence of every $j$-depth request has been defined,
we define the characteristic sequence of a $(j+1)$-depth request $r_i:=(u_{i}, \ell_{i})$
to be  characteristic sequence of $r_{u_i}$ concatenated with the term $i$.
\end{defi}
%

The notion of Definition \ref{iAEntPEpUs} can be transferred to 
strings in $S:=\cup_{i<k} S_i$, which we shall refer to as \emph{codes}.

\begin{defi}[Successor and predecessor codes]
Given a solution  $\tuple{S_i, i<k}$ to a layered KC-sequence $\tuple{(u_i,\ell_i), i<k}$
and two strings $\sigma,\tau\in S:=\cup_{i<k} S_i$, we say that $\sigma$ is
an immediate predecessor of $\tau$ (and that $\tau$ is an immediate successor of 
$\sigma$) if $\sigma\preceq\tau$ and there exist $j<t$ such that 
$\sigma\in S_j, \tau\in S_t$ and $(u_j,\ell_j)$ is an immediate predecessor of $(u_t,\ell_t)$.
If $\sigma\in S_i$ then we say that the index of $\sigma$ is $i$.
\end{defi}

\subsection{The greedy solution to a layered KC-sequence}\label{AGapDZxJHM}
The solution $\tuple{S_i,i<k}$ for Theorem \ref{vOPTnu8Le8} will be composed out of the 
greedy solutions of auxiliary plain (relativized) KC-sequences. 
In particular, each string $\sigma$ that is enumerated into some $S_i$,
corresponds to a plain KC-sequence $L_{\sigma}$ relative to $\sigma$ 
in the sense of \S \ref{f9OYHFR8ER}.
The idea is that any strings enumerated into $S:=\cup_j S_j$ for the satisfaction of a request
whose immediate predecessor is $(u_i,\ell_i)$, will be chosen by the greedy algorithm corresponding to 
$L_{\sigma}$ for some $\sigma\in S_i$.
Recall that $\wgt{L_{\sigma}}$ denotes the weight of $L_{\sigma}$. 
Table \ref{ta:vschelparbouatd} displays the main parameters of the greedy solution. 
\begin{defi}[Clear extensions]
Given a set $S$ of strings and strings $\sigma\preceq \tau$,
we say that $\tau$ is an $S$-clear extension of $\sigma$ if
there are no extensions of $\sigma$ in $S$ which are $\preceq$-comparable to
$\tau$.
\end{defi}
By the analysis in \S \ref{6Azke7oXF} and the relativization in \S \ref{f9OYHFR8ER}
we get the following fact.
\begin{equation}\label{XfvqHWyaU}
\parbox{12cm}{Given a KC-sequence $L=\tuple{\ell_i, i<k}$ relative to $\sigma$ 
and its greedy solution $\tuple{\sigma_i\ |\ i<k}$, let
$S=\{\sigma_i\ |\ i<k\}$. Then there exists an $S$-clear extension of $\sigma$ 
of length $\ell$ if and only if $\wgt{L}\leq 2^{-|\sigma|}-2^{-\ell}$.}
\end{equation}
In the solution $\tuple{S_i,i<k}$ that we construct, we view the sets
$S_i$ as {\em ordered sets of strings},  where order is given by the arrival time
of each string that is enumerated into $S_i$.
\begin{defi}[Arrival ordering in $S_i$]
Given the greedy solution $\tuple{S_i,i<k}$ to a layered KC-sequence
and some $j<k$, consider $\eta,\tau\in S_j$. We write $\eta<\tau$ if
$\eta$ was enumerated into $S_j$ at an earlier stage than $\tau$,
\ie if there exists $t<k$ such that $\eta\in S_j[t]$ and 
$\tau\not\in S_j[t]$. The terms `earliest' or `latest' string in $S_i$
refer to the minimal and maximal elements of $S_i$ with respect to this ordering.
\end{defi}

In the following, for each $t>0$ and each string $\sigma$, 
we let $L_{\sigma}[t]$ denote the state of the request set $L_{\sigma}$
at the end of stage $t$, \ie when the definition of the greedy solution of 
$\btuple{(u_{i}, \ell_{i}), i<t}$ has been completed. Then at the next step,
the definition of the greedy solution of $\btuple{(u_{i}, \ell_{i}), i<t+1}$ will be given
as an extension of the greedy solution of $\btuple{(u_{i}, \ell_{i}), i<t}$
(\ie by enumerating into the sets $S_i$ and extending the set sequence with $S_t$).
Additional requests will also be enumerated into the sets $L_{\sigma}$, thus
determining  $L_{\sigma}[t+1]$. 
Definition \ref{NrJFDTH4f} is an induction on all $k\in\Nat$ with $k>0$.
For notational simplicity, in the induction step
for the definition of  $\tuple{S_i[k+1],i<k+1}$, we write $S_i$ for $S_i[k+1]$,
$S'_i$ for $S_i[k]$ and $L_{\sigma}$ for 
$L_{\sigma}[k+1]$.



\begin{table}
\begin{center}
  \begin{tabular}{rlrl}
 \hline\hline\\[-0.2cm]
{\small $(u_i,\ell_i)$:}   &  {\small $i$th layered request}	&\hspace{0.5cm}
{\small $S$:}   & {\small codes in $\cup_{i<k} S_i$}\\[1ex]
{\small $S_i$:}	  &  {\small satisfaction set for $\btuple{u_i,\ell_i}$} &\hspace{0.5cm}
{\small $L_{\sigma}$:}    & {\small plain KC-sequence corresponding to $\sigma\in S$}\\[1ex]
\hline\hline
\end{tabular}
\caption{Parameters of the greedy solution to a layered KC-sequence}
\label{ta:vschelparbouatd}
\end{center}
\end{table}

\begin{defi}[Greedy solution for layered KC-sequences]\label{NrJFDTH4f}
Let $S_0=\{\lambda\}$ and $L_{\sigma}[0]=\emptyset$ for all $\sigma$. This specifies the greedy solution to the sequence $\langle (\ast,0) \rangle$. 
The greedy solution  $\tuple{S_i,i<k+1}$ of $\btuple{(u_{i}, \ell_{i}), i<k+1}$
is obtained by extending the sets in the greedy solution $\tuple{S'_i,i<k}$ of
$\btuple{(u_{i}, \ell_{i}), i<k}$ with at most one string each, and concatenating the
modified sequence $\tuple{S_i,i<k}$ with a singleton $S_{k}$ as follows.
Let $S'=\cup_{i<k} S'_i$ and assume that $L_{\sigma}[k]$ is defined for all $\sigma$.

Let  $\tuple{v_j,j<t}$ be the characteristic sequence of the latest term $(u_{k}, \ell_{k})$. 
For every $x<k$ such that $x\neq v_j$ for all $j<t$, we define $S_x=S'_x$.
\begin{equation}\label{A9LwkOv6ys}
\textrm{\underline{Hypothesis}}:\left\{\  \parbox{12cm}{\textup{there 
exists some $j<t-1$ and  $\sigma\in S'_{v_j}$
with $\wgtb{L_{\sigma}[k]}\leq 2^{-|\sigma|}-2^{-\ell_{v_{j+1}}}$}
\textup{[equivalently, $\sigma$ has an $S'$-clear extension of length $\ell_{v_{j+1}}$]}.}\right\}
\end{equation}
Let $j_0$ be the largest number $j$ satisfying \eqref{A9LwkOv6ys} and
let  $\sigma_{j_0}$ be the earliest such string $\sigma\in S'_{v_{{j_0}}}$.
For each $j\in [j_0,t-2]$, starting from $j=j_0$,

\begin{enumerate}[\hspace{0.5cm}(a)]
\item enumerate the plain KC-request $\ell_{v_{j+1}}$ into $L_{\sigma_{j}}$;
\item let $\sigma_{j+1}$ be the string given to request $\ell_{v_{j+1}}$ of $L_{\sigma_{j}}$
by the greedy KC-solution relative to $\sigma_{j}$;
\item if $j<t-2$, define $S_{v_{j+1}}=S'_{v_{j+1}}\cup\{\sigma_{j+1}\}$, and
if $j=t-2$ define $S_{k}=S_{v_{t-1}}=\{\sigma_{t-1}\}$. 
\end{enumerate}
Also define 
$S_{v_{x}}=S'_{v_{x}}$ for all $x\leq j_0$.
String $\sigma_{j_0}$ is called {\tt the base} of stage $k+1$ of the greedy solution.
\end{defi}

\begin{rem}[Basic properties]\label{4YFaADyxRF}
We note
the following properties that are direct consequences of the construction.
\begin{enumerate}[\hspace{0.2cm}(i)]
\item If $d_1$ is the depth of request $(u_{k},\ell_{k})$
and $d_0$ is the depth of the base $\sigma_{j_0}$ of stage $k+1$,
then $d_0<d_1$ and for each $d\in (d_0, d_1]$ exactly one code of depth $d$ 
is enumerated into $S$.
\item the only code $\sigma\in S[k]$ for which $L_{\sigma}[k]\neq L_{\sigma}[k+1]$
is the base of stage $k+1$.
\item if $L_{\sigma}[k]\neq L_{\sigma}[k+1]$ for some code $\sigma$ 
of depth $d$ then there is an immediate successor of $\sigma$ (hence
a code of depth $d+1$) in 
$S[k+1]-S[k]$.
\end{enumerate}
\end{rem}

\subsection{Verification of the layered greedy solution}\label{ZMVKSMqTVQ}
Note that subject to 
\eqref{A9LwkOv6ys}  holding for the duration of the definition of
the greedy solution in Definition \ref{NrJFDTH4f},
the algorithm given satisfies each
layered request, producing a solution according to Definition \ref{jPdiYuWBqz}.
In particular, the prefix-freeness condition is met by the properties of the
(relativized) greedy solutions to the plain KC-sequences $L_{\sigma}$,
according to the analysis in \S \ref{6Azke7oXF} and \S \ref{f9OYHFR8ER}.
Hence it suffices to show that 
\eqref{A9LwkOv6ys} holds at the beginning of each stage of the induction in
 Definition \ref{jPdiYuWBqz}. 
We first establish a  monotonicity property
regarding the traces of strings enumerated successively into $S_i$.

\begin{lem}[Monotonicity of traces]\label{g36rcysNtq}
Let $i\in\Nat$ and let $\eta_0,\eta_1\in S_i$
with $\eta_0<\eta_1$. 
Then every `1' in the trace of $L_{\eta_0}$ is
to the right of (\ie at a larger position than) each `1' in the trace of $L_{\eta_1}$.
\end{lem}
\begin{proof}
The proof is by induction on stages. At stage $0$, 
the claim holds trivially. Suppose that it holds at the end of stage $k$ and  $\eta_0,\eta_1\in S_i[k+1]$. 
Since $\eta_0,\eta_1\in S_i[k+1]$, the codes $\eta_0,\eta_1$ have the same depth
and the same length. 
At stage $k+1$, if either of
$L_{\eta_0}$, $L_{\eta_1}$ changes, then, by Remark \ref{4YFaADyxRF}, exactly one of 
$L_{\eta_0}$, $L_{\eta_1}$ changes and 
 one of the following holds: 
\begin{enumerate}[\hspace{0.5cm}(a)]
\item $\eta_0\in S[k]$ and $\eta_1\not\in S[k]$;
\item $\eta_0\in S[k]$ and $\eta_1\in S[k]$;
\end{enumerate}
The characteristic sequence of the new request $(u_{k}, \ell_{k})$
has a unique term $j$ such that  $(u_{j}, \ell_{j})$ is an immediate successor 
to  $(u_{i}, \ell_{i})$ (otherwise neither of $L_{\eta_0}$, $L_{\eta_1}$ would change at stage $k+1$).

First assume that (a) holds.
We claim that $\wgt{L_{\eta_0}}[k]> 2^{-|\eta_0|}-2^{-\ell_j}$: if this were not the case there would be
a clear extension of $\eta_0$ of length $\ell_j$, so the greedy solution would choose to satisfy
$(u_{j}, \ell_{j})$ above $\eta_0$ and not above $\eta_1$ (when it chooses the maximum index satisfying \eqref{A9LwkOv6ys}).
By the plain KC-theorem analysis,
it follows that there are no 1s in the trace of $L_{\eta_0}$ up to position $\ell_j$. On the other hand, 
the first stage where $L_{\eta_1}\neq \emptyset$ is $k+1$, so the trace of 
$L_{\eta_1}$ at the end of stage $k+1$ has a single 1 which is at position $\ell_j$.
This establishes the induction hypothesis for this case. 

In case (b),  the base of stage $k+1$ is
either $\eta_0$ or $\eta_1$.
We then subdivide further into two cases.  If  \eqref{A9LwkOv6ys} holds for $\eta_0$ at stage $k+1$,  then since earlier strings are given preference in choosing the base, it follows that $\eta_0$ is the base at stage $k+1$, 
 $L_{\eta_0}[k]\neq L_{\eta_0}[k+1]$ and 
$L_{\eta_1}[k]=L_{\eta_1}[k+1]$. By the induction hypothesis we may let 
$\ell$ be such that
all 1s in the trace of $L_{\eta_0}[k]$ are to the right of 
position $\ell$ and all 1s in the trace of $L_{\eta_1}[k]$
are strictly to the left of position $\ell$.  Since  \eqref{A9LwkOv6ys} holds for $\eta_0$,  we 
have $\ell_j\geq \ell$. After the enumeration of a request of length $\ell_j$ in
$L_{\eta_0}$, the required monotonicity property will hold.
If  \eqref{A9LwkOv6ys} does not hold for $\eta_0$, then let $\ell$ be defined in the same way. 
In this case we
have $\ell_j< \ell$. The request $(u_j,\ell_j)$ (as specified above) will be satisfied above $\eta_1$, and since all the 1s in the trace of 
$L_{\eta_1}[k]$ are strictly to the left of position $\ell$, the induction step follows.
\end{proof}

With Lemma \ref{g36rcysNtq} in place, it is not difficult to complete our verification of the layered greedy solution. It suffices to prove the result for layered KC-sequences of finite depth (i.e.\ for which there exists $d$ such that all requests are of depth at most $d$), since if (\ref{A9LwkOv6ys}) fails then it does so at some finite stage. The proof for sequences of finite depth $d$ then proceeds by induction on $d$. 

The case for $d=1$ is just the plain Kraft-Chaitin theorem, so suppose the result holds for $d\geq 1$.  Given a layered KC-sequence $L=\tuple{(u_i,\ell_i), i<k_1}$ of depth $d+1$ and weight at most 1, for which (\ref{A9LwkOv6ys}) fails to hold at stage $k_1$, we produce a sequence $L'=\tuple{(u_i',\ell_i'), i<k_0}$ of depth $d$, which is also of weight at most 1 and for which  (\ref{A9LwkOv6ys}) fails at stage $k_0$, contradicting the induction hypothesis.  In order to enumerate the sequence $L'$, we run the greedy solution for $L$ and act as follows during each stage $i$. At stage $i$, let $r(i)$ be the least $j$ such that we have not yet enumerated the $j$th request $(u_j',\ell_j')$ into $L'$ (with $r(0)=0$). 

\begin{itemize} 
\item If $(u_i,\ell_i)$ is of depth at most $d$, then let $(u_{r(i)}',\ell_{r(i)}')=(r(u_i),\ell_i)$ and enumerate it into $L'$, i.e.\ we enumerate essentially the same request into $L'$, but have to adjust the index of the request pointed to, in order to take account of the different rates of enumeration into $L$ and $L'$.
This is the $r(i)$th request enumerated into $L'$ and is the element of $L'$ which \emph{corresponds to} the $i$th element    $(u_i,\ell_i)$ of $L$. 
\item  If $(u_i,\ell_i)$ is of depth $d+1$ and the base at stage $i+1$ is of depth $d$, then make no enumeration into $L'$ at this stage. 
\item If $(u_i,\ell_i)$ is of depth $d+1$ and the base at stage $i+1$ is of depth $<d$ then let $\tuple{v_j, j\leq d+1}$ be the characteristic sequence of $(u_i,\ell_i)$, so that $(u_{v_d},\ell_{v_d})$ is the immediate predecessor of $(u_i,\ell_i)$ in $L$, while  $(u_{v_{d-1}},\ell_{v_{d-1}})$ is the immediate predecessor of $(u_{v_d},\ell_{v_d})$. Enumerate the request $(r(v_{d-1}),\ell_{v_d})$ into $L'$. So this is a request of depth $d$, which points to the request in $L'$ corresponding to  $(u_{v_{d-1}},\ell_{v_{d-1}})$ in $L$. For future reference, we say that this request  $(r(v_{d-1}),\ell_{v_d})$ which we have just enumerated into  $L'$ is a \emph{secondary} request and is a \emph{brother} of the $r(v_d)$th request enumerated into $L'$, which is the request corresponding to $(u_{v_d},\ell_{v_d})$. 
\end{itemize} 

With $L'$ enumerated as above, it then follows directly by induction on the stage $i$, that 
\begin{equation*}
\parbox{13cm}{any $\sigma$ is a code of depth $d'\leq d$ by the end of stage $i$ in the greedy solution for 
$L$ iff it is a code of the same depth in the greedy solution for $L'$ by the end of stage $r(i)$.}
\end{equation*}
By the first clause of this equivalence we mean
that $\sigma \in \cup_{j<i} S_i[i]$ 
and has at most $d$ proper predecessors in this set.
So if (\ref{A9LwkOv6ys}) fails at stage $k_1$ in the layered greedy solution for $L$, then it also fails at stage $k_0=r(k_1)$ in the layered greedy solution for $L'$. It remains to show that the total weight of $L'$ is at most 1.  For this it suffices to show that 
\begin{equation*}
\parbox{11cm}{the total weight of all of the secondary requests in $L'$ is at most the total weight of all the $(d+1)$-depth requests in $L$. }
\end{equation*}
Note that secondary requests in $L'$ are  $d$-depth and each has a unique brother which is of $d$-depth and \emph{primary}, i.e.\ not secondary.
So suppose that the $i_0$th request $(u_{i_0}',\ell_{i_0}')$ enumerated into $L'$ is primary and of depth $d$, and let $i_1<\cdots<i_m$ be all of the indices of secondary requests which are brothers of that primary request. Let $k+1$ be the stage in the layered greedy solution for $L$ at which we enumerate the $i_m$th request into $L'$, so that $r(k+1)=i_m$ and the $k$th request $(u_{k},\ell_k)$ in $L$ is of depth $d+1$. Let $i$ be such that the $i$th request in $L$ corresponds to the $i_0$th request in $L'$. Then at the end of stage $k$ in the layered greedy solution for $L$, there exist $m$ strings in $S_i[k]$, and these can be indexed $\{ \sigma_{i_0}, \sigma_{i_1}, \dots, \sigma_{i_{m-1}} \}$ in such a way that each $\sigma_{i_j}$ is made a code at stage $i_j$ in the layered greedy solution for $L'$, in order to satisfy the $i_j$th request in $L'$. From Lemma \ref{g36rcysNtq} applied to $\{ \sigma_{i_0}, \sigma_{i_1}, \dots, \sigma_{i_{m-1}} \}$ and the fact that the base is of depth $<d$ at stage $k+1$ in the layered greedy solution for $L$, it follows that by the end of stage $k+1$ the total weight of all of the $(d+1)$-depth requests pointing to the $i$th request $(u_i,\ell_i)$ is at least $m2^{-\ell_i}$. So this is at least the total weight of all the secondary requests with brother the  $i_0$th request $(u_{i_0}',\ell_{i_0}')=(u_{i_0}',\ell_i)$ in $L'$, as required.

\section{Proving  our main results}\label{pm7mWNnuPr}
In the previous section we gave a rather general method for coding
sequences into binary streams, via the concept of layered KC-requests. 
Here we wish to use this general coding tool in order to compress
an arbitrary binary stream, to a degree that matches a given information content measure.

\subsection{From layered requests to code-trees}
We may view the greedy solution of a layered KC-sequence as a \ce tree of strings,
a code-tree\footnote{A partial map $\sigma\mapsto V_{\sigma}$ from strings to sets of strings is called computably enumerable (c.e.)
if the family of sets $\tuple{V_{\sigma}}$ is uniformly computably enumerable. We say  $\sigma\mapsto V_{\sigma}$
is {\em monotone} if for each $\sigma$, $V_{\sigma}$ is prefix-free and contains only proper extensions of $\sigma$.
We define the code-tree $S$ generated by
a monotone \ce map $\sigma\mapsto V_{\sigma}$ from strings to sets of strings,
inductively as follows. 
The empty string $\lambda$ is in $S$ and has depth 0.
If $\sigma\in S$ and has depth $\ell$, then all strings in $V_{\sigma}$ are in $S$ and have depth $\ell+1$. 
A set of strings $S$ is a code-tree  if it is the code-tree 
generated  by some monotone \ce  map $\sigma\mapsto V_{\sigma}$.
If $\sigma, \tau\in S$  then $\sigma$ is the immediate predecessor of $\tau$ (and $\tau$ is an immediate  successor of $\sigma$) if $\tau\in V_{\sigma}$.   We say $\sigma$ is an
 ancestor of $\tau$ if there exists a sequence $\sigma_i, i\leq k$ such that 
 $\sigma_0=\sigma$, $\sigma_k=\tau$ and $\sigma_{j+1}\in V_{\sigma_j}$ for each $j<k$.
 By an infinite path through $S$, we mean a stream with infinitely many initial segments in $S$. 
 If $S$ is a code-tree, then we may refer to an element of $S$ as a code.} 
 in which each node is effectively mapped to some string, in a monotone
way. Of course, such a mapping was only implicit in \S \ref{AGapDZxJHM}
(where strings in $S$ encode the characteristic sequences of layered requests)
but will be made explicit here. In the following, layered KC-sequences will be
of the form $\tuple{(\sigma_i,u_i,\ell_i),\ i<k}$, so that each request is augmented by
some string. The significance of this is that each request codes a certain string which is
determined at the enumeration of the request. As a consequence, if $S=\tuple{S_i,\ i<k}$ is the
greedy solution to $\tuple{(\sigma_i,u_i,\ell_i),\ i<k}$, then each string in $S_i$ is a code
for $\sigma_i$.

In the following it will be useful to be able to produce
code-trees that contain codes which are not prefixed by any string in some \pf
set of strings $Q$ of sufficiently small weight ($Q$ might be a member of a universal Martin-L\"{o}f test, for example).
This task can also be achieved through the greedy solution of a suitable
sequence of layered requests.
To this end it 
will be convenient work with a `slowed down' or `filtered' enumeration of 
$Q$. During the construction, we therefore enumerate a \pf set of strings $D$, which  contains
those strings in the code-tree which are prefixed by strings in $Q$.
Here we define the enumeration of $D$, given effective enumerations of $Q$ and any code-tree 
$S$.

An element $\sigma$ of a set of strings is a leaf of that set if there are no proper extensions of 
$\sigma$ belonging to the set.  
A tree-enumeration $\tuple{S_s}$ of a code-tree $S$ 
is one where for any $\sigma\in S_s$, all ancestors of $\sigma$ in $S$ are already in $S_s$.
Note that in this case, a leaf of $S_s$ is not necessarily a leaf of $S$ or even $S_{s+1}$.
Given a tree-enumeration $\tuple{S_s}$ and a \ce
set $Q$ of `forbidden strings', 
in the following definition we 
define an enumeration $(D_s)$ of a set $D$ of strings with the property that
$\dbra{D_s}\subseteq\dbra{Q_s}$ (\ie any stream with a prefix in $D_s$ has a prefix in $Q_s$)
and if $\sigma$ is enumerated into $D$ at stage $s$ then $\sigma$ is a leaf of $S_s$.

\begin{defi}[Filtered enumeration of $Q$]\label{PvKfvvp7j8}
Given a code-tree $S$ with computable tree-enumeration $\tuple{S_s}$
and a \ce \pf set of strings $Q$ with enumeration $\tuple{Q_s}$, we define 
$\tuple{D_s}$ inductively. 
\begin{itemize}
\item At stage 0 let $D_0=\emptyset$;
\item At stage $s+1$, if there exists a leaf of $S_{s}$ 
which does not belong to $D_s$
and has a prefix in $Q_s$, pick the most recently enumerated into $S$ such leaf and  enumerate it into $D$.
\end{itemize}
A stage $s+1$ is called {\em expansionary} if $D_s=D_{s+1}$. 
Otherwise $s+1$ is called an {\em adaptive} stage.
\end{defi}

Since any string enumerated in $D_s$ has a prefix in $Q_s$ we have 
$\dbra{D_s}\subseteq\dbra{Q_s}$, while the converse is not generally true.
So $\tuple{D_s}$ is a filtered version of $\tuple{Q_{s}}$, 
in the sense that only existing leaf-codes that are currently prefixed
by a string in $Q_s$ can be enumerated into $D_s$. 
Definition \ref{PvKfvvp7j8} defines the enumeration of $D$ at stage $s+1$ in terms of the enumerations
that have occurred in $Q,S$ in the previous stages up to $s$.
Hence Definition \ref{PvKfvvp7j8} can be used recursively
{\em during} the construction of a code-tree $S$, so that the enumeration in $S$ at stage $s+1$ may depend
on $D_s$, which itself is defined in terms of $S_t, t\leq s$.
In \S\ref{oONPU21TN}
we shall enumerate $S$ in such a way that no new string 
will be enumerated into $S_{s+1}$ extending any string in $D_s$. 
This means that the strings in $D$ will actually be leaves of $S$, and that $D$ will be a prefix-free set.

\subsection{Proof of Theorem \ref{yMWDPEuwPx}}\label{oONPU21TN}
Recall the statement of Theorem \ref{yMWDPEuwPx}:
\begin{equation*}
\parbox{14cm}{If $I$ is any partial computable information content measure, then
every binary stream $X$ such that $\forall n\ I(X\restr_n)\de$ 
can be coded into a \ml random binary stream $Y$ such that $X$ is computable
from $Y$ with oracle-use $n\mapsto \min_{i\geq n} I(X\restr_i)$.}
\end{equation*}
Let $I$ be as in the hypothesis of the theorem, \ie 
a partial computable function such that $\sum_{I(\sigma)\de} 2^{-I(\sigma)}$ is finite.
Note that if $X$ is computable from $Y$ with oracle-use $g$, then for each integer $c$
there exists some $Z$ such that $X$ is computable from $Z$ with oracle-use $n\mapsto g(n)-c$.
Hence without loss of generality we may assume that
$\sum_{\sigma} 2^{-I(\sigma)}<1$. Then 
Theorem \ref{yMWDPEuwPx} follows from the following technical lemma, for a suitable choice of
a \ce \pf set $Q$ of strings.
\begin{lem}\label{EwCVSBQo7C}
If $I$ is a computable information content measure and $Q$ is a \pf set of strings such that
$\sum_{\sigma} 2^{-I(\sigma)}+\sum_{\sigma\in Q} 2^{-|\sigma|}<1$, then
every binary stream $X$ can be coded into a binary stream $Y$ such that $X$ is computable
from $Y$ with oracle-use $n\mapsto \min_{i\geq n} I(X\restr_i)$ and $Y$ does not have a 
prefix in $Q$.
\end{lem}

In order to obtain Theorem \ref{yMWDPEuwPx}, for each $c$ 
consider the set $Q_c=\{\sigma\ |\ K(\sigma)<|\sigma|-c\}$ of the strings which can be compressed by at least $c$ bits.
By the counting theorem from \cite{MR0411829} it follows that there exists a constant $d_0$ such that for all $c$ we have
$\mu(\dbra{Q_c})< 2^{d_0-c}$. Hence given an information content measure $I$ with $\sum_{\sigma} 2^{-I(\sigma)}<1$
we may choose some $c$ such that $\sum_{\sigma} 2^{-I(\sigma)}+\mu(\dbra{Q_c})<1$. Since $Q_c$ is \ce there exists
a \ce \pf set $Q$ such that $\dbra{Q}=\dbra{Q_c}$, \ie the two sets have the same infinite extensions. Then
$\mu(\dbra{Q_c})=\sum_{\sigma\in Q} 2^{-|\sigma|}$ so the hypothesis of Lemma \ref{EwCVSBQo7C} holds for $Q$.
Moreover by the definition of $Q_c,Q$ it follows that any binary stream without a prefix in $Q$ is \ml random.
In this way, Theorem \ref{yMWDPEuwPx} is a consequence of  Lemma \ref{EwCVSBQo7C}  for this particular set $Q$.

The following partial ordering on strings will guide the pointers in the definition of our
layered KC-sequence $L$.
\begin{defi}[The $I$-ordering]\label{7m54knuyG9}
We define $\sigma$ to be {\em the $I$-predecessor of $\tau$} if 
$I(\tau\restr_i)\de$ for all $i\leq |\tau|$,
$\sigma$ is a proper prefix of $\tau$ and $\sigma$ is the largest prefix of $\tau$
such that $I(\sigma)<I(\tau)$.
\end{defi}
We first define a layered request set $L$ based on $I$, 
and later extend it to $L'$ which produces $Q$-avoiding codes (\ie codes that are not prefixed
by strings in $Q$).
Let $(\tau_s)$ be 
an effective list of strings such that 
if $I(\tau_j)\de$ then $I(\rho)\de$ for all $\rho\prec\tau_j$ and the strings 
$\tau_i, i<j$ include all proper prefixes of $\tau_j$. 

{\em Definition of $L=\tuple{(\tau_i,u_i,\ell_i)\ |\ i\in\Nat}$.}
For each $s$, if $\tau_s$ has an $I$-predecessor, enumerate a
request for $\tau_s$ of length $I(\tau_s)$ pointing to the request of the $I$-predecessor 
of $\tau_s$, \ie a request $r_s=(\tau_s, u_s,I(\tau_s))$ 
where $u_s$ is the index of the predecessor of
$\tau_s$.
If $\tau_s$ does not have an $I$-predecessor, then enumerate
a request  $r_s=(\tau_s, 0,I(\tau_s))$.

We will now interweave the requests in
$L$ with additional requests based on $Q$, and form a request sequence
$L'= \tuple{r_i'[s]}$, where $r_i'=(\tau'_i,u_i',\ell_i')$. 
The mechanism by which these extra requests are inserted works as follows. 
All requests in $L$ are initially \emph{unejected}. 
At each stage $s+1$ we consider the greedy solution $S_i, i\leq s$ that has been generated
for the requests $r_i', i\leq s$, and the corresponding filtered enumeration $D_i, i\leq s+1$
with respect to $S_i, i\leq s$ according to Definition \ref{PvKfvvp7j8}.
At each 
expansionary stage we consider the least unejected request from $L$, and we then eject this request and enumerate it into $L'$ (with indices modified so as to reflect the different numbering of requests in $L$ and $L'$.) At each adaptive stage we shall not eject any requests from $L$, but will rather enumerate a new request directly into $L'$ -- this request can be thought of as a copy of some previous request which now has to be satisfied again due an enumeration into $D_{s+1}$.
Each request in $L$ will have a {\em current $L'$-index} which may be redefined during the 
construction at most finitely many times. The $L'$-{\em index of a code} $\sigma\in S:=\cup_i S_i$ is the unique $i$ such that $\sigma \in S_i$, and the $L$ -index of $\sigma$ is then the unique $j$ such that $i$ is the current $L'$-index for $r_j\in L$.

\begin{defi}[Definition of $L'$ given $L, Q$]\label{dUlAq3Dz7}
We define $L'=\tuple{(\tau_i', u_i',\ell_i')\ |\ i\in\Nat}$ in stages as follows.
If stage $s+1$ is adaptive, then consider the unique code
$\sigma\in D_{s+1}-D_s$, let $i$ be its $L'$-index, let $j$ be its $L$-index,  and enumerate
$r_{s+1}':=(\tau_i', u_i', \ell_i')$ into $L'$, setting $\tau_{s+1}'=\tau_i'$, $u_{s+1}'=u_i'$, 
$\ell_{s+1}'=\ell_i'$; in this case the current $L'$-index of
the $L$-request $r_j$ is redefined to be $s+1$, and we say that $r_i'$ becomes {\em outdated}.
If stage $s+1$ is expansionary, eject the least unejected request $r_i$ in $L$
and enumerate 
$r_{s+1}':=(\tau_i, u_{j}',\ell_i)$ into $L'$, where $j$ is the current $L'$-index of the $L$-request
$r_{u_i}$; in this case the current $L'$-index of $r_i$ is defined to be $s+1$
and we also set $\tau_{s+1}'=\tau_i$, $u_{s+1}'=u_j'$, $\ell_{s+1}'=\ell_i$.
\end{defi}

Note that if  stage $s+1$ is adaptive, then there exists a leaf of $S_{s}$ 
which does not belong to $D_s$
and has a prefix in $Q_s$, and that, according to Definition \ref{PvKfvvp7j8}, 
we pick that most recently enumerated into $S$ and  enumerate it into $D$. 
According to Definition \ref{dUlAq3Dz7}, we then enumerate a new leaf of the 
same length into $S$ at this stage. It follows that one can only ever have finitely many adaptive stages in a row. 

In the following, by $\wgt{H}$ for a set of strings $H$, we mean $\sum_{\sigma\in H} 2^{-|\sigma|}$.
\begin{lem}\label{uch2OtHyt}
Given any layered KC-sequence $L$ and any \pf set of strings $Q$, the weight
of the layered KC-sequence $L'$ of Definition \ref{dUlAq3Dz7} is 
bounded by $\wgt{L}+\wgt{Q}$.
\end{lem}
\begin{proof}
According to the remarks after Definition \ref{PvKfvvp7j8}
we have that $D:=\cup_s D_s$ is \pf and
$\dbra{D}\subseteq\dbra{Q}$.
Moreover 
\[
\wgt{L'}=\wgt{L}+\wgt{D}\leq \wgt{L}+\wgt{Q}
\]
since each request in $L'$ is either a request ejected from $L$ or
a request of the same weight as the current enumeration into $D$.
\end{proof}
By Lemma \ref{uch2OtHyt} and
the hypothesis of Lemma \ref{EwCVSBQo7C}
we have that $\wgt{L'}<1$.
Hence $L'$ is a valid layered KC-sequence and the code set $S$
that is generated during the construction of $L'$ is the greedy solution
of $L'$.


We may now show that given any $X$ such that $I(X\restr_n)\de$ for all $n$,
there exists $Y$ which computes $X$ with oracle-use 
$\min_{i\geq n} I(X\restr_i)$. 
Consider
the lengths $(n_i)$ of the prefixes of $X$ that are local $I$-minima
in the following sense: $X\restr_{n_0}$ is the longest prefix of $X$ such that 
$I(X\restr_{n_0})=\min_n I(X\restr_n)$; inductively, $X\restr_{n_{i+1}}$ is the longest
prefix of $X$ which is of length greater than $n_i$ and such that 
$I(X\restr_{n_{i+1}})=\min_{n>n_i} I(X\restr_n)$.

Then consider the set of indices $J$ of the $L'$-requests
which are never outdated and which correspond to the strings $X\restr_{n_i}$ (\ie which have $X\restr_{n_i}$ as
their first coordinate) and let 
$S_X=\cup_{i\in J} S_i$. Clearly $S_X$ is infinite, so let $Y$ be a stream with infinitely many
prefixes from $S_X$. 

Note that $I(X\restr_{n_i})<I(X\restr_{n_{i+1}})$ for each $i$
and each $Y\restr_{I(X\restr_{n_i})}$ is a code in  the greedy solution $S$, which is not
prefixed by any string in $Q$ and which can be
effectively decoded into the segment $X\restr_{n_i}$.
Therefore $X$ is computable from $Y$ with oracle-use $n\mapsto \min_{i\geq n} I(X\restr_{i})$ 
as follows. Let $m_i=I(X\restr_{n_i})$.
Given $n>0$ in order to compute $X\restr_n$, we first compute 
the least $i$ such that $n\leq n_i$, 
using only $Y\restr_{m_{i}}$
of the oracle $Y$. Then we can use the given oracle Turing machine
in order to compute $X\restr_{n_{i}}$, and therefore $X\restr_n$, from $Y$ with oracle-use $m_{i}$.
If we let $n_{-1}=-1$ then 
by definition we have $m_{i}=\min_{t> n_{i-1}} I(X\restr_t)\leq \min_{t\geq n} I(X\restr_t)$
which concludes the proof of Lemma \ref{EwCVSBQo7C}.

\subsection{Proof of Corollary \ref{FxNsrDIzhj}}
Recall the statement of Corollary \ref{FxNsrDIzhj}:
\begin{equation*}
\parbox{14cm}{If $g$ is a computable upper 
bound on the initial segment \pf complexity of a stream $X$,
then $X$ is computable from a \ml random stream $Y$ with oracle-use 
$n\mapsto \min_{i\geq n} g(i)$.}
\end{equation*}
Corollary \ref{FxNsrDIzhj} is a direct consequence of Theorem \ref{yMWDPEuwPx}
and the following lemma.
\begin{lem}
Given a stream $X$ and a computable upper bound $g$ on $n\mapsto K(X\restr_n)$,
there exists a partial computable information content measure $I$ such that 
$I(X\restr_n)=g(n)$ for all $n$.  
\end{lem}
\begin{proof}
Given $g$ and a computable monotone approximation $(K_s)$ to $\sigma\mapsto K(\sigma)$
we define $I$ as follows: for each $\sigma$ wait until a stage $s$ such that
$K_s(\sigma)\leq g(|\sigma|)$. If and when such a stage appears, define $I(\sigma)=g(|\sigma|)$.
Clearly $I$ is partial computable.
It remains to show that $I$ is an information content measure.
Note that each definition $I(\sigma)\de$ made in our construction, corresponds to
a unique description of $\tau_{\sigma}$ of 
$\sigma$ with respect to the universal \pf machine $U$ (namely the first
description of length at most $g(|\sigma|)$). Therefore
\[
\sum_{I(\sigma)\de} 2^{-I(\sigma)}= 
\sum_{I(\sigma)\de} 2^{-g(|\sigma|)}= 
\sum_{I(\sigma)\de} 2^{-|\tau_{\sigma}|}\leq 
\sum_{U(\rho)\de} 2^{-|\rho|}<1
\]
which concludes the proof.
\end{proof}

\subsection{Proof of Theorem \ref{Z6AicbrzJm}}
Recall the statement of Theorem \ref{Z6AicbrzJm}:
\begin{equation*}
\parbox{14cm}{every binary stream $X$ can be coded into a \ml 
random binary stream $Y$ such that $X$ is computable
from $Y$ with oracle-use $n\mapsto \min_{i\geq n} K(X\restr_i)+\log n$.}
\end{equation*}

We need the following technical lemma.
\begin{lem}\label{pXeAC7F56T}
There exists a constant $c$ such that, for each $\sigma$:
\begin{equation}\label{Vn38AR6Axh}
\sum_{\rho\succeq \sigma} 2^{-K(\rho)-\log |\rho|}\leq 2^{-K(\sigma)+c}. \ \ 
\end{equation}
\end{lem}
\begin{proof}
By the maximality of $\sigma\mapsto 2^{-K(\sigma)}$ as a \lce semi-measure,
it suffices to show that the map which sends $\sigma$ to the left-hand-side expression of
\eqref{Vn38AR6Axh} (which is clearly left-c.e.) is a semi-measure. We have
\[
\sum_{\sigma}\sum_{\sigma\preceq\rho} 2^{-K(\rho)-\log(|\rho|)}=
\sum_{\rho} |\rho|\cdot 2^{-K(\rho)-\log(|\rho|)}=
\sum_{\rho} 2^{-K(\rho)}<1,
\]
which concludes the proof of the lemma.
\end{proof}

By Lemma \ref{pXeAC7F56T}
we may consider an enumeration $(U_s)$ of the underlying universal machine $U$
such that
\begin{equation}\label{Fs2AmRqtO}
\sum_{\rho\succeq \sigma} 2^{-K_s(\rho)-\log |\rho|-c}\leq 2^{-K_s(\sigma)} \ \ 
\textrm{for each $\sigma$ and each stage $s$,} 
\end{equation}
where $K_s(\sigma)$ is the \pf complexity of $\sigma$ with respect to $U_s$. 
We may also assume that at each stage $s+1$ there exists exactly one string $\sigma$
such that $K_{s+1}(\sigma)<K_{s}(\sigma)$.  

It suffices to prove the following technical lemma, where $c$ is the constant from (\ref{Fs2AmRqtO}). It is convenient to assume in what follows that $c\geq 1$. 
\begin{lem}\label{ADVptCOxwy}
If $Q$ is a \pf set of strings such that
$\sum_{\sigma} 2^{-K(\sigma)}+\sum_{\sigma\in Q} 2^{-|\sigma|}<1$, then
every binary stream $X$ can be coded into a binary stream $Y$ such that $X$ is computable
from $Y$ with oracle-use $n\mapsto \min_{i\geq n} (K(X\restr_i) +\log i)+c$ and $Y$ does not have a 
prefix in $Q$.
\end{lem}

We follow the form of the argument given in \S \ref{oONPU21TN}. 
The main difference here is that we cannot directly define an analogue of the ordering of strings
in Definition \ref{7m54knuyG9}
based on the measure $\sigma\mapsto  K(\sigma) +\log |\sigma|+c$ 
(indicating which segments of each string are coded) since we only have an approximation
of our measure at each stage.
Instead, such an ordering will be implicitly approximated, in a coding process where
false approximations to   $\sigma\mapsto K(\sigma)$ correspond to suboptimal codes or requests.
 At each stage of the construction we identify the unique request that needs updating, and
the request that the replacement has to point to.

\begin{defi}[Target and pre-target]
The target of stage $s+1$ is the unique
string $\sigma$ such that $K_{s+1}(\sigma)<K_s(\sigma)$.
The longest initial segment $\tau$ of the target $\sigma$ at stage $s+1$, such that
$K_{s+1}(\tau)+ \log |\tau| <K_{s+1}(\sigma) + \log |\sigma| $, if this exists, is called the pre-target at stage $s+1$.
\end{defi}

We also need to identify requests that have previously pointed to the most recent suboptimal request (for the target)
and update them.
We express the subsequence of descendants of a request\footnote{The request $r'$ is a descendant of the request $r$ (relative to a given layered KC-request sequence $L$) if $r'=r$ or there is a sequence $r=r_0, r_1, \dots, r_m=r'$ of members of $L$ such that each $r_{i+1}$ is an immediate successor of $r_i$ for $0\leq i<m$. In this case we also say that $r$ is an ancestor of $r'$.}  in a layered KC-sequence 
as another layered KC-sequence by a suitable
manipulation of the indices. First though, we need to define a notion of \emph{validity} for requests. 

\begin{defi}[Validity of requests]
A layered KC-request $r=(\sigma,u,\ell)$ is valid (relative to the layered KC-sequence of which it is a member) at stage $s$ if $\ell=K_{s}(\sigma)+\log |\sigma|+c$
and for every ancestor $(\sigma',u',\ell')$ of $r$ we have $\ell'=K_{s}(\sigma')+\log |\sigma'|+c$.
If $\rho$ is a string, a layered KC-request $(\sigma,u,\ell)$
is a $\rho$-request if $\sigma=\rho$, and may also be referred to as a request for $\rho$.
\end{defi}

\begin{defi}[Subtrees of a layered KC-sequence]\label{4UXMUoeDp1}
Given a stage $s$, a finite layered KC-sequence $L_k=\tuple{r_i=(\sigma_i,u_i,\ell_i)\ |\ i<k}$, and $t<k$, let 
$\tuple{r_{n_i}\ |\ i<n^{\ast}}$
be the list of descendants of $r_t$ in $L_k$ which are valid at stage $s$, labelled such that $n_i<n_{i+1}$ for all $i$.
The  $r_t[s]$-subtree in $L_k$ is the  layered KC-sequence
$\tuple{r'_{i}=(\sigma_{n_i},u^{\ast}_{i},\ell_{n_i})\ |\ i<n^{\ast}}$, where
$u^{\ast}_0=\ast$ and
for each $i<n^{\ast}$, 
$u^{\ast}_i$ is the unique $j$ such that $u_{n_i}=n_j$.
\end{defi}

%
We may now formalise the manner in which a given layered KC-sequence $L_k$ 
may be extended by appending a  subtree above the last element. The cloning operation
is given in the form that will be used in the construction. Note that the last request in the given
layered KC-sequence is not part of the given subtree. This is because in the construction, this last element
will be the replacement for the target request at that stage.

\begin{defi}[Clone extension of layered KC-sequences]\label{GFbCFxP43}
Given a stage $s$,  a finite layered KC-sequence $L_{k+1}=\tuple{r_i\ |\ i<k+1}$ and $t<k$,
let $L_{k}=\tuple{r_i\ |\ i<k}$
and let $\tuple{r'_i=(\sigma'_i, u'_i, \ell'_i)\ |\ i<n^{\ast}}$
be the $r_t[s]$-subtree  in $L_k$.
We define the $r_t[s]$-clone extension of $L_{k+1}$ to be the layered KC-sequence
$L= \tuple{r^{\star}_i\ |\ i < k+n^{\ast}}$
where:
\begin{itemize}
\item for each $i< k+1$ we have $r^{\star}_i=r_i$;
\item for each $0< i<n^{\ast}$, we have $r^{\star}_{k+i}=(\sigma'_{i}, u'_{i}+k, \ell'_{i})$.
\end{itemize}
\end{defi}
Note that the clone extension of a layered KC-sequence, as defined in Definition 
\ref{GFbCFxP43}, produces another layered KC-sequence.

The operation described in Definition \ref{GFbCFxP43} can now be used
in order to define the universal layered KC-sequence $L$ inductively, by
successive concatenations of a finite
layered KC-sequence  $L_s$ that is defined at each stage. Here we use $\ast$
for indicating concatenation on layered KC-sequences. 
 In the following definition, we deviate from the previous implicit convention that $L_k$ is a sequence of length $k$. 

\begin{defi}[Universal layered KC-sequence]\label{heSec6EtwP}
At stage $s+1$ suppose that $L_s$ has been defined, let $\sigma$ be the target
and define
$L'=L_s\ast (\sigma,u,K_{s+1}(\sigma)+\log |\sigma|+c)$ where $u$ is 
the index of the valid request for the pre-target, provided that the latter exists,
and $0$ otherwise. If there is no $\sigma$-request in $L_s$, define
$L_{s+1}=L'$; otherwise
let $r_t$ be the $\sigma$-request which was valid in $L_s$ (at stage $s$) and 
define $L_{s+1}$ to be the $r_t[s]$-clone extension of $L'$.
Finally define $L=\lim_s L_s$.
\end{defi}
In the following, given a \pf machine $U$,
we let $\wgt{U}$ denote the weight of the domain of $U$.
 \begin{lem}
The weight of the universal layered KC-sequence of Definition \ref{heSec6EtwP}
is at most $\wgt{U}$.
\end{lem}
\begin{proof}
It suffices to show that at each stage $s+1$, the weight of the layered KC-sequence
$L_s$ of Definition \ref{heSec6EtwP} is bounded above by the weight of
the domain of $U_s$. At each stage $s+1$, the target $\sigma$ receives a new shorter
description in $U$, and  
there are two kinds of requests that
are added to $L$:
\begin{itemize}
\item the request $(\sigma,u,K_{s+1}(\sigma)+\log |\sigma|+c)$ corresponding to the target;
\item the $r_t[s]$-subtree of $L_s$, in case $\sigma$ already has a $\sigma$-request.
\end{itemize}
The first request has weight $2^{-K_{s+1}(\sigma)-\log |\sigma|-c}$
and by \eqref{Fs2AmRqtO} the weight of the  $r_t[s]$-subtree of $L_s$
is bounded above by $2^{-K_s(\sigma)}=2^{-K_{s+1}(\sigma)-1}$. Overall, recalling the assumption that $c\geq 1$, the
increase in the weight of $L$ at stage $s+1$ is bounded above by
\[
2^{-K_{s+1}(\sigma)-\log |\sigma|-c}+2^{-K_{s+1}(\sigma)-1}\leq 2^{-K_{s+1}(\sigma)}.
\]
On the other hand,
the increase in $\wgt{U}$ at stage $s+1$ is $2^{-K_{s+1}(\sigma)}$.
Hence inductively, $\wgt{L_s}\leq \wgt{U_s}$ for each $s$, which means that 
$\wgt{L}\leq \wgt{U}$.
\end{proof}

Since the weight of the layered KC-sequence of Definition \ref{heSec6EtwP} is bounded by 1, 
we may consider its greedy solution. 
Since  we also want codes that avoid the given set $Q$ of Lemma \ref{ADVptCOxwy}, however, 
we first need to obtain a modified layered KC-sequence $L'$, exactly as we did in 
\S \ref{oONPU21TN}. Given $L$ and $Q$ we generate $L'$ as defined in
Definition \ref{dUlAq3Dz7}.
By Lemma \ref{uch2OtHyt}
we have that $\wgt{L'}<1$, so we may consider the greedy solution $S'$ of $L'$.
Note that some requests $L'$ will become {\em outdated} through enumerations into $Q$, while
some will become {\em invalid} (i.e.\ not valid) through the compression of strings. By the construction of $L$ and $L'$ 
it follows that no infinite chain of requests, such that each request points to the previous request in the chain, contains any outdated or invalid requests.

Given any $X$ we will now construct a suitable $Y$ from the code-tree $S'$.
We can think of the
layered KC-sequence $L$ or $L'$ as coding certain segments of $X\restr_{n_i}$
of $X$.

\begin{defi}\label{IpaMBFHcH4}
The \emph{significant} initial segments of a real $Z$ are $Z\restr_{n_i}$, $i\in\Nat$ where 
$(n_i)$ is defined inductively by $n_0=\arg\min_i (K(Z\restr_i) + \log i)$ and 
$n_{t+1}=\arg\min_{t> n_i} (K(Z\restr_i) + \log i)$.
\end{defi}

If $X\restr_{n_i}$ are the significant initial segments of $X$, then by the definition of $L$ it follows that
for each $X\restr_{n_i}$ there exists a $X\restr_{n_i}$-request of length $K(X\restr_{n_i})+\log n_i+c$.
Let $S_X'$ be the set of all codes in $S'$ which satisfy any of the above requests of  $L'$.
Since $S'_X$ is infinite, there exists an infinite path $Y$ through it, and a corresponding infinite chain of requests such that each points to the previous request in the sequence, meaning that none become outdated or invalid.  So $Y$ does not have
a prefix in $Q$, and is the union of the codes $Y\restr_{m_i}$, where 
$m_i=K(X\restr_{n_i})+\log n_i +c$ for
each $i$.
Furthermore, $Y\restr_{m_i}$, uniformly computes $X\restr_{n_i}$ for each $i$.
From this fact, it follows  that
$X$ is computable from $Y$ with 
 oracle-use $n\mapsto \min_{i\geq n} (K(X\restr_i)+\log i) +c$.

\section{Concluding remarks and open problems}\label{7vgsUJbYd}
We have shown formally that
every stream $X$ can be coded into an algorithmically random code-stream $Y$, from
which it is effectively recoverable with oracle-use the information content of $X$, as measured by the \pf
initial segment complexity of $X$, up to $\log n$. Moreover we noted that this oracle-use is optimal, up to $3\log n$.
The main breakthrough in this work is the elimination of an overhead of $\sqrt{n}\cdot \log n$ in the oracle use,
which exists in all previous approaches to this problem, independently of the complexity of the source $X$.
As we discussed in \S\ref{xbu9Sgt5wv}, 
this overhead is inherent to all previous coding methods, and is overwhelming both
from the point of view of compressible sources with initial segment complexity  $X$ is $\smos{\sqrt{n}\cdot \log n}$
as well as compared to our logarithmic overhead. 
In the case of computable information content measure, or computable upper bounds in the initial segment complexity,
Corollary \ref{FxNsrDIzhj} gives overhead 0 (\ie oracle use exactly the given upper bound) while any previous method
retains the overhead $\sqrt{n}\cdot \log n$.

The second half of our contribution is conceptual and methodological, in terms of a new general coding method,
which was necessary for the elimination of the bottleneck, intrinsic in
all previous approaches, $\sqrt{n}\cdot \log n$. 
This was presented in a fully general form in \ref{iKCR6nBLqB},
and can be seen as an infinitary analogue of the classic 
Kraft-McMillan and Huffman tools 
for the construction of prefix codes with minimum redundancy. As such, we expect that our method will have
further applications on problems where sequences of messages are needed to be coded into an online stream,
without suffering an accumulation of the overheads that the words of a \pf code inherently have, in an open-ended
code-stream.
A detailed comparison of our method with all the existing methods was conducted in \S\ref{xbu9Sgt5wv},
and a single characteristic property was isolated, which is present in all of the previous approaches but absent in our method.
We concluded that it is the liberation from this restrictive property \eqref{ba23HDmLNi} that allows our method to achieve
optimal coding and break through the bottleneck that is characteristic in all previous approaches.

There are two main ways that our work can be improved and extended.

{\bf Tightness of the upper bounds on the oracle-use.}
In view of the oracle-use bound $K(X\restr_n)+\log n$ in 
Theorem \ref{Z6AicbrzJm}, a natural question is if the stronger upper bound
$K(X\restr_n)$ is possible (or $\min_{i\geq n} K(X\restr_i)$).
After all, the worst-case oblivious bounds of
\cite{optcod16,bakugaopt}  that we discussed in \S\ref{xbu9Sgt5wv} are tight in a rather absolute way,
even up to $\log\log\log n$ differences.
Moreover 
it was shown in  \cite{koba_rod} that 
the oracle-use bound $\min_{i\geq n} K(X\restr_i)$ is achievable in the special case
of \lce reals. Despite these examples,
we have reasons to conjecture that
such an ultra-tight upper bound is not achievable in general.
\begin{conjecture}
There exists $X$ such that in any oracle computation of $X$ by any $Y$, the oracle-use is not bounded
above by $n\mapsto K(X\restr_n)$.
\end{conjecture}
The intuition here is based on the non-uniformity of initial segment complexity, which was
discussed in the beginning of \S \ref{wrsYJDwjf}.
Roughly speaking, and in the context of the arguments in \S\ref{pm7mWNnuPr}, 
in any coding of $X$ into $Y$ with oracle-use $n\mapsto K(X\restr_n)$,
a change in the approximation to $K(X\restr_n)$ 
would render all codes of segments of $X$ that
are longer than $n$, sub-optimal; at the same time, this change need not
affect $K(X\restr_i)$, $i>n$.
It is this non-monotonicity for the settling times of $n\mapsto K(X\restr_n)$ (rather than the 
non-monotonicity of $K$ as a function) that seems to be the obstacle to such a tight bound on the
oracle-use.

{\bf Feasibility of coding.}  
Our online algorithm produces approximations of the code stream of a source, based on the current
computations of the universal compression machine. In other words, at each stage $s$ the code stream of the source
is incompressible with respect to the universal computations that have appeared up to stage $s$; moreover the code stream
reaches a limit as $s\to\infty$.
Since the limit code stream is required to be incompressible against the computations of any Turing machine,
the asymptotic outcome cannot be determined effectively. The analogue of this observation in the compression of finite strings
is the fact that shortest programs for finite strings cannot be computed effectively.

It would be interesting to study ways in which our coding method can be used more effectively.
In the case of finite strings, such approaches include \cite{Bauens6597753,Teutsch:2014:687474} where
it was shown that given any string, it is possible to construct in polynomial time a list of programs that is guaranteed 
to contain a description of the given string, whose length is within $\bigo{1}$ of its Kolmogorov complexity.
A version of the above result 
in a randomized setting, where we allow a small error probability and
the use a few random bits, 
was obtained in \cite{BauwensZimand2014LLS}. 
In the more relevant case of source coding into random streams subject to time and space resource bounds,
the previous methods of Ku\v{c}era \cite{MR820784}, 
G\'{a}cs \cite{MR859105}, Ryabko \cite{Rya84,Rya86} were successfully 
adapted, with some additional work, to many  resource-bounded settings in 
Doty \cite{cie/Doty06,Doty08}. Moreover, Balc{\'a}zar,  Gavald{\`a} and  
Hermo in \cite{MR1455141} 
showed a special case of our Corollary \ref{FxNsrDIzhj} 
for streams of logarithmic initial segment complexity
and with time and space resource bounds.
In this fashion, and given that our methods
are rather different to the methods used in the above articles,
it would be interesting to investigate the extent to which
our main results,
Theorems \ref{yMWDPEuwPx} and \ref{Z6AicbrzJm}, as well as Corollary \ref{FxNsrDIzhj},
hold in the presence of computational feasibility restrictions.

\bibliographystyle{abbrv}
\bibliography{kcode}
\end{document}